\documentclass[journal]{IEEEtran}
\ifCLASSINFOpdf
 \usepackage[pdftex]{graphicx}
\else
\fi
\usepackage{amsmath, amsfonts, amsthm, calc}
\usepackage{algorithmic}

\usepackage{array}

\ifCLASSOPTIONcompsoc
 \usepackage[caption=false,font=normalsize,labelfont=sf,textfont=sf,subrefformat=parens,labelformat=parens]{subfig}
\else
 \usepackage[caption=false,font=footnotesize,subrefformat=parens,labelformat=parens]{subfig}
\fi
\usepackage[hyphens]{url}

\usepackage{bm, multicol, eufrak}
\newtheorem{theorem}{Theorem}
\newtheorem{proposition}{Proposition}

\newcommand{\E}{\text{E}}
\newcommand{\q}{\text{q}}
\renewcommand{\P}{\text{P}}
\renewcommand{\v}{\text{vec}}

\begin{document}
%
\title{Anomaly Detection in Partially Observed Traffic Networks}
%
%
%

\author{Elizabeth~Hou,~\IEEEmembership{Student Member,~IEEE,}
	Yasin~Y{\i}lmaz,~\IEEEmembership{Member,~IEEE,} and~Alfred~O. Hero,~\IEEEmembership{Fellow,~IEEE}
	\thanks{Elizabeth Hou and Alfred Hero are with the EECS Department at University of Michigan, Ann Arbor, MI, Contact email: emhou@umich.edu }
	\thanks{Yasin Y{\i}lmaz was with the EECS Department at the University of Michigan, Ann Arbor, MI. He is now with the Electrical Engineering Department at the University of South Florida, Tampa, FL.}}

\maketitle

\begin{abstract}
This paper addresses the problem of detecting anomalous activity in traffic networks where the network is not directly observed. Given knowledge of what the node-to-node traffic in a network should be, any activity that differs significantly from this baseline would be considered anomalous. We propose a Bayesian hierarchical model for estimating the traffic rates and detecting anomalous changes in the network. The probabilistic nature of the model allows us to perform statistical goodness-of-fit tests to detect significant deviations from a baseline network. We show that due to the more defined structure of the hierarchical Bayesian model, such tests perform well even when the empirical models estimated by the EM algorithm are misspecified. We apply our model to both simulated and real datasets to demonstrate its superior performance over existing alternatives.
\end{abstract}


%
\IEEEpeerreviewmaketitle

\section{Introduction}

In today's connected world, communication is increasingly voluminous, diverse, and essential. Phone calls, delivery services, and the Internet are all modern amenities that send massive amounts of traffic over immense networks. Thus network security, such as the ability to detect network intrusions or illegal network activity, plays a vital role in defending these network infrastructures. For example, (i) computer networks can protect themselves from malware such as botnets by identifying unusual network flow patterns; (ii) supply chains can prevent cargo theft by monitoring the schedule of shipments or out-of-route journeys between warehouses; (iii) law enforcement agencies can uncover smuggling operations by detecting alternative modes of transporting goods.

Identifying unusual network activity requires a good estimator of the true network traffic, including the anomalous activity, in order to distinguish it from a baseline of what the network should look like. However, often it is not possible to observe the network directly due to constraints such as cost, protocols, or legal restrictions. This makes the problem of estimating the rate of traffic between nodes in a network difficult because the edges between nodes are latent unobserved variables. Network tomography approaches have been previously proposed for estimating network topology or reconstructing link traffic from incomplete measurements and limited knowledge about network connectivity. However for network anomography, the detection of anomalous deviations of traffic in the network, highly accurate estimation of all network traffic may not be necessary. It often suffices to detect perturbations within the network at an aggregate or global scale.  This paper addresses the problem of network anomography rather than that of network tomography or traffic estimation.

\subsection{Related Work} 

Broadly defined, the network tomography problem is to reconstruct complete network properties, e.g., source-destination (SD) traffic or network topology,  based on incomplete data. The term ``network tomography" was introduced in \cite{Vardi1996} where the objective is to estimate unknown source destination traffic intensities given observations of link traffic and known network topology. Since the publication of \cite{Vardi1996}, the scope of the term network tomography has been used in a much broader sense (see the review papers \cite{coates2002internet, Medina:2002:TME:633025.633041, castro2004network}, and \cite{lawrence2006network}). For example, a variety of passive or active packet probing strategies have been used for topology reconstruction of the Internet, including unicast, multicast, or multi-multicast \cite{coates2002maximum, caceres1999multicast}, and \cite{rabbat2004multiple}; or using different statistical measures including packet loss, packet delay, or correlation \cite{tsang2003network, shih2003unicast, shih2007hierarchical}, and \cite{duffield2006network}. 

In the formulation of \cite{Vardi1996}, the network tomography objective is to determine the total amount of traffic between SD pairs given knowledge of the physical network topology and the total amount of traffic flowing over links, called the link data. This leads to the linear model for the observations $ \bm{y}^t = \bm{A} \bm{x}^t $ where $ \bm{A} $ is the known routing matrix defining the routing paths, and at each time point $t$, $ \bm{y}^t $ is a vector of the observed total traffic on the links and $ \bm{x}^t $ is a vector of the unobserved message traffic between SD pairs. Using the model that the elements of $ \bm{x}^t $ are independent and Poisson distributed, an expectation-maximization (EM) maximum likelihood estimator (MLE) and a method of moments estimator are proposed in \cite{Vardi1996} for the Poisson rate parameters $ \bm{\lambda} $. The authors of \cite{Tebaldi} propose a Bayesian conditionally Poisson model, which uses a Markov chain Monte Carlo (MCMC) method to iteratively draw samples from the joint posterior of $ \bm{\lambda} $ and $\bm{x}$. The authors of \cite{Cao00ascalable} and \cite{Cao00time-varyingnetwork} assume the message traffic is instead from a Normal distribution, obtaining a computationally simpler estimator of the SD traffic rates. The authors of \cite{8101008} relax the assumption that the traffic is an independent and identically Poisson distributed sequence and instead consider the network as a directly observable Markov chain. Under this weaker assumption, they derive a threshold estimator for the Hoeffding test in order to detect if the network contains anomalous activity.

In \cite{Vanderbei94anem} the authors propose an EM approach for Poisson maximum likelihood estimation when the network topology is unknown; however, their solution is only computationally feasible for very small networks and it does not account for observations of traffic through interior nodes. This has led to simpler and more scalable solutions in the form of gravity models where the rate of traffic between each SD pair is modeled by $ x_{sd} = (N_s N_d) / N $ where $ N_s $ and $ N_d $ are the total traffic out of the source node and into the destination node respectively and $ N $ is the total traffic in the network. Standard gravity models do not account for the interior nodes, thus in \cite{Zhang:2003:FAC:781027.781053} and \cite{Zhang:2003:IAT:863955.863990} tomogravity and entropy regularized tomogravity models were proposed, which incorporate the interior node information in the second stage of their algorithm. The authors of \cite{6058636} generalize the tomogravity model from a rank one (time periods are independent) to a low rank approximation (time periods are correlated) and allow additional observations on individual SD pairs. Similarly, the authors of \cite{6497613} and \cite{7098434} use a low rank model with network traffic maps to incorporate a sparse anomaly matrix, and they solve their multiple convex objectives with the alternating direction method of multipliers (ADMM) algorithm. 

Dimensionality reduction has also been used directly for anomaly detection in the SD traffic flows in networks. Under the assumption that traffic links have low rank structure, the authors in \cite{lakhina2004characterization} and \cite{lakhina2004diagnosing} use Principle Component Analysis (PCA) to separate the anomalous traffic from the nominal traffic. This low rank framework is generalized to applying PCA in networks that are temporally low rank or have dynamic routing matrices, in \cite{zhang2005network}. The authors of \cite{zhang2005network} also coin the term ``network anomography" to reflect the influence of network topology reconstruction, which is a necessary component to detecting anomalies in a network with unknown structure. However, later work in \cite{Ringberg:2007} discusses the limitations of PCA for detecting anomalous network traffic, e.g., it is sensitive to (i) the choice of subspace size; (ii) the way traffic measurements are aggregated; (iii) large anomalies. The low rank plus sparse framework is extended to online setting with a subspace tracking algorithm in \cite{7536642}.

Specifically for Internet Protocol (IP) networks, some works prefer to perform anomaly detection on the flows from the IP packets instead of the SD flows. The authors of \cite{li2006detection} use PCA to separate the anomalous and nominal flows from sketches (random aggregations of IP flows) while the authors of \cite{krishnamurthy2003sketch} model the sketches as time series and detect change points with forecasting. The works of \cite{thottan2003anomaly} and \cite{gu2005detecting} also perform change point detection using windowed hypothesis testing with generalized likelihood ratio or relative entropy respectively.

Because our approach in this paper is based on traffic networks or SD models, these types of approaches were the focus of our related works subsection. However, networks can also be represented as graph models or as features of the network characteristics. This subsection would be incomplete if it did not mention anomaly detection approaches to other types of network models. So, we refer to some survey papers that cover many of the recent techniques in graph based approaches: \cite{ranshous2015anomaly} and \cite{akoglu2015graph}. In particular, similar to the low rank approaches for SD networks, there are low rank approaches to graph models such as \cite{egilmez2014spectral} who assume the inverse covariance matrix of their wireless sensor network data has a graph structure and solve a low rank penalized Gaussian graphical model problem and \cite{8323201} who impose graph smoothness by a low rank assumption on graph Laplacian of the features of the network. \cite{lee2013anomaly} also uses a low rank approach on their KDD intrusion data set, but they directly apply the low rank assumption to the network characteristics of their data.

\subsection{Our Contribution}

In this paper, we consider networks where an exterior node (a node in an SD pair) only transmits and receives messages from a few other nodes, but because we cannot observe the network directly, we do not know which SD pairs have traffic and which do not. Thus, we develop a novel framework to detect anomalous traffic in sparse networks with unknown sparsity pattern. Our contributions are the following. 1) In order to estimate the network traffic, we propose a parametric hierarchical model that alternates between estimating the unobserved network traffic and optimizing for the best fit rates of traffic using the EM algorithm. 2) We warm-start the algorithm with the solution to non-parametric minimum relative entropy model that directly projects the rates of traffic onto the nearest attainable sparse network. 3) Since we do not make assumptions of fixed edge structure in our model, it allows us to accommodate the possibility of anomalous edges in the actual network structure because anomalies will never be known in advance. 4) Using our probabilistic model's estimator of actual traffic rates, we test for anomalous network activity by comparing it to a baseline to determine which deviations are anomalies and which are estimation noise. We develop specific statistical tests, based on the generalized likelihood ratio framework, to control for the false positive rate of our probabilistic model, and show that even when our models are misspecified, our tests can accurately detect anomalous activity in the network.

The rest of the paper is organized in the following way. Section II proposes a problem formulation of the network we are interested in and our assumptions about it. Section III describes our proposed hierarchical Bayesian model, which is solved with a generalized EM algorithm and warm-starting the EM with a solution that satisfies the minimum relative entropy principle. Section IV describes our anomaly detection scheme through statistical goodness of fit tests and Section V describes the computational complexity of our method. Section VI contains simulation results of the performance of our proposed estimators and applications to the CTU-13 dataset of botnet traffic and a dataset of NYC taxicab traffic. Finally, Section VII concludes the paper.

\section{Proposed Formulation} 

We give a simple diagram of a notional network in Fig.~\subref*{fig:oracle}. An exterior node, $V_i$, sends messages, $N^t_{ij}$, at a rate, $\Lambda_{ij}$, to another exterior node, $V_j$, at each time point, $t$. Messages can flow through interior nodes, such as $U_1$, but the interior nodes do not absorb or create messages. Because the magnitude of flow is just the total number of messages that have been sent from one node to another, network traffic between nodes is a counting process. For tractability, it is common to assume the messages are independent and identically distributed (i.i.d.) and the total number of messages in a time period is from some parametric distribution. The Poisson distribution is the most natural choice because it models events occurring independently with a constant rate, and it is used by \cite{Vardi1996}, \cite{Vanderbei94anem}, \cite{Tebaldi}, \cite{Cao00ascalable}, and \cite{Cao00time-varyingnetwork} although the latter two works use a Normal approximation to the Poisson for additional tractability. Under these Poisson process assumptions, the uniformly minimum variance unbiased estimator is simply the maximum likelihood estimator (MLE). 

However, this is a very strong and unrealistic assumption because it would require being able to track every single message being passed in the network. Thus, we are interested in the much weaker assumption that we can only monitor the nodes themselves. Fig.~\subref*{fig:actual} shows what we can actually observe from the network under this weaker assumption. While we also observe the total amount of traffic, unlike in \cite{Vardi1996}, we do not know the network topology.

Since we can only monitor the nodes, we can only observe the total ingress and egress of the exterior nodes. Thus we know an exterior node, $V_i$, transmits $N^t_{i\cdot}$ messages and receives $N^t_{\cdot i}$ messages, but we do not know which of the other nodes it is interacting with. We can also observe the flow through interior nodes, but we cannot distinguish where the messages come from or are going to. For instance, in Fig.~\subref*{fig:oracle}, an interior node, such as $U_1$, will observe all messages, $F_1^t = N^t_{14} + N^t_{2P} $, that flow through it, but it will not be able to distinguish the number of messages from each SD pair or whether all the SD pairs actually send messages.

\begin{figure}[h]
	\centering
	\subfloat[Proposed Network: $V_i$ - exterior nodes, $U_i$ - interior nodes, $N^t_{ij}$ - messages from node $i$ to node $j$ at time point $t$]{\includegraphics[width=\linewidth]{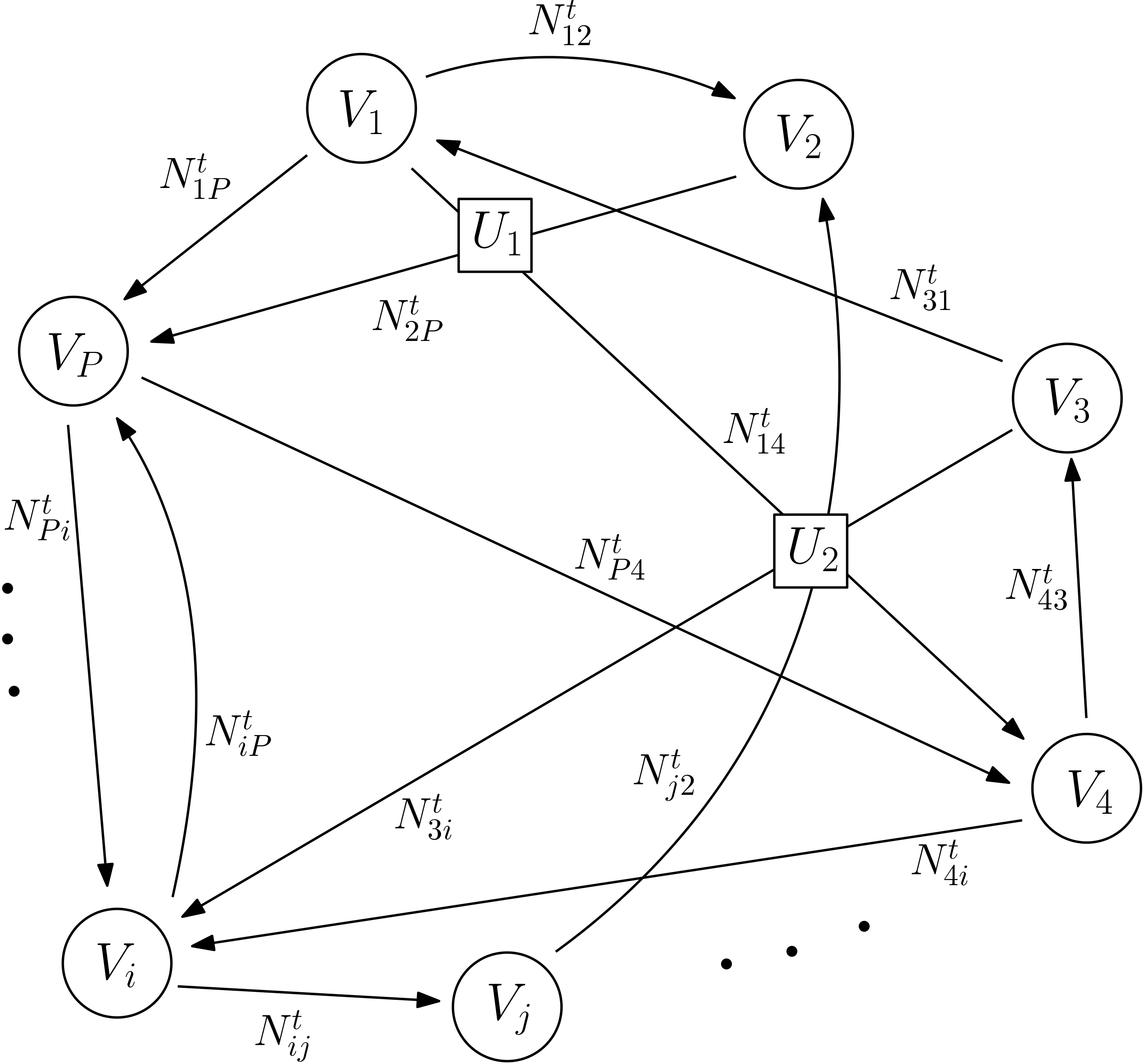}\label{fig:oracle}}\hfil
	\subfloat[Actual Observed Network: $N^t_{i\cdot}$ - total egress of exterior nodes, $N^t_{\cdot i}$ - total ingress of exterior nodes, $F^t_i$ - total flow through interior nodes]{\includegraphics[width=\linewidth]{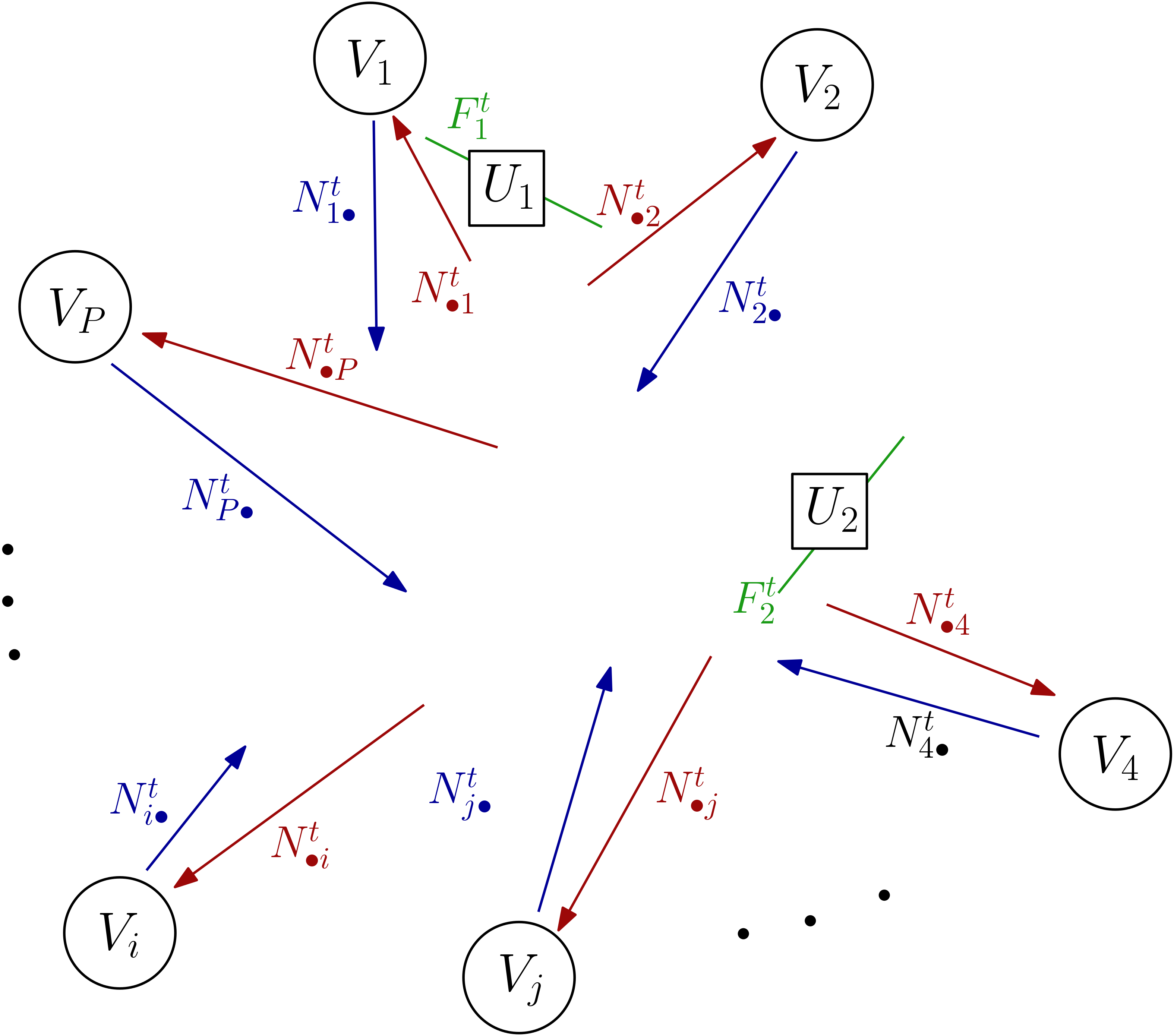}\label{fig:actual}}
	\caption{Diagram of a network with $P$ exterior nodes and 2 interior nodes.}
	\label{fig:problem}
\end{figure}

A network with $P$ exterior nodes can naturally be mathematically formulated as a $P \times P$ matrix, which is observed $T$ times. Let $\bm{N}^t$ be the unobserved traffic matrix at time instance $t$ and let the elements of the matrix, $N^t_{ij}$, be the amount of traffic between nodes $i$ and $j$. The row and column sums of the traffic are denoted by $\bm{R} = [ N_{1.} \ldots N_{P.} ]' $ and $\bm{C} = [ N_{.1} \ldots N_{.P} ]'$ respectively, and $ \bm{F} = [ F_{h} ]$ are the observed flows through interior nodes, which are indexed by $h$. The traffic at each time instance $t$ is generated from a distribution with mean $ \bm{\Lambda} $, the true intensity/rate parameter of the matrix, and $ \bm{\Lambda}_0 $ is the baseline parameter of a network without any anomalies. This mathematical formulation is shown below.

\begin{minipage}{.4\linewidth}
	\begin{equation*} 
	\hspace{-10pt} \bm{N}^t \hspace{-3pt} = \hspace{-3pt} \left[ \begin{array}{ccccc}
	0 & N^t_{12} & N^t_{13} & \cdots & N^t_{1P} \\
	N^t_{21} & 0 & N^t_{23} & \cdots & N^t_{2P} \\
	N^t_{31} & N^t_{32} & 0 & \cdots & N^t_{3P} \\
	\vdots & \vdots & \vdots & \ddots & \vdots \\
	N^t_{P1} & N^t_{P2} & N^t_{P3} & \cdots & 0
	\end{array} \right]
	\end{equation*}
\end{minipage} \hspace{45pt} 
\begin{minipage}{.4\linewidth}
	\centering
	\smallskip
	\underline{Observations}
	\smallskip
		\begin{itemize}
			\item[] $ N^t_{.j} = \sum_{i=1}^P N^t_{ij} $
			\item[] $ N^t_{i.} = \sum_{j=1}^P N^t_{ij} $
			\item[] $ F^t_{h} = \sum N^t_{ij} $
			\item[] $ N^t_{ij} \text{ for some } ij $
		\end{itemize}

\end{minipage} \\

We assume \textit{a priori} that the distribution of the rate matrix is centered around some baseline rate matrix $\bm{\Lambda_0}$, which are the assumed rates when there is no anomalous activity. We then update this prior distribution using the observations $\mathcal{D} = \{\bm{R}^t, \bm{C}^t, \bm{F}^t \}_{t=1}^T$ in order to get a distribution of the rates $ \P(\bm{\Lambda} | \mathcal{D}) $, which does account for potential anomalous activity.

\section{Hierarchical Poisson Model with EM}

We propose a generative model that assumes a series of statistical distributions govern the generation of the network. We assume that the messages $N^t_{ij}$ passed through the network are Poisson distributed with rates $\Lambda_{ij}$. However, because we cannot observe the traffic network directly, we do not have the complete Poisson likelihood and use the EM algorithm. In the following subsections, we will show a series of generative models with increasing complexity that attain successively higher accuracy. Then we will discuss warm-starting the EM algorithm at a robust initial solution to compensate for its sensitivity to initialization. 

\subsection{Proposed Hierarchical Bayesian Model}

\subsubsection{Maximum Likelihood by EM}

The simplest hierarchical model assumes all priors are uniform, thus the only distributional assumption is that likelihood $ \P(\bm{N}^1, \dots, \bm{N}^T| \bm{\Lambda}) $ is $\prod_{t=1}^T \prod_{ij} Poisson(\Lambda_{ij}) $. The maximum likelihood estimator for the Poisson rates $\bm{\Lambda}$ can be approximated by lower bounds of the observed likelihood $ \P(\mathcal{D}| \bm{\Lambda}) $ using the maximum likelihood expectation maximization (MLEM) algorithm. The MLEM alternates between computing a lower bound on the likelihood function $ \P(\mathcal{D}| \bm{\Lambda})$, the E-step, and maximizing the lower bound, the M-step. A general expression for the E-step bound can be expressed as:  
\begin{flalign} \label{MLE_low_bound}
\hspace{-3pt} \log \P(\mathcal{D}| \bm{\Lambda}) \hspace{-1pt} \geq \hspace{-1pt} \sum_{t=1}^T \hspace{-1pt} \E_{\q^t} \hspace{-1pt} \left(\log \P(\bm{R}^t, \bm{C}^t, \bm{F}^t, \bm{N}^t | \bm{\Lambda}) \right) \hspace{-1pt} + \text{H}(\q^t) \hspace{-1pt}
\end{flalign}
where $\q^t$ is an arbitrarily chosen distribution of $\bm{N}^t $, $\E_{\q^t}$ denotes statistical expectation with respect to the reference distribution $\q^t$, and  $\text{H}(\q^t)$ is the Shannon entropy of $\q^t$. The choice of $\q^t$ that makes the bound \eqref{MLE_low_bound} the tightest, and results in the fastest convergence of the MLEM algorithm, is $ \q^t = \P(\bm{N}^t| \bm{R}^t, \bm{C}^t, \bm{F}^t, \bm{\Lambda}), $ (see Section 11.4.7 of \cite{murphy2012machine}); however, this is not a tractable distribution. When the observations consist of the row and column sums of the matrix $\bm{N}^t$, this distribution is the multivariate Fisher's noncentral hypergeometric distribution, and when the flows are also observed the distribution is unknown. Unfortunately, use of this optimal distribution leads to an intractable E-step in the MLEM algorithm due to the coupling (dependence) between the row and column sums of $\bm{N}^t$. As an alternative we can weaken the bound on the likelihood function \eqref{MLE_low_bound} by using a different distribution $\q$ that leads to an easier E-step. To this aim, we propose to use a distribution $\q$ that decouples the row sum from the column sum; equivalent to assuming that each sum is independent, e.g., as if each were computed with different realizations of $\bm{N}^t$.

\begin{proposition} \label{prop:indep_time_LB}
Assume $t_{1}, t_{2}$ and $t_{3}$ are different time points so that observations at these time points are independent 
\begin{flalign*}
\emph{P}(\mathcal{D}| \bm{\Lambda}) = \prod_{t_{1}=1}^T \emph{P}(\bm{R}^{t_{1}}| \bm{\Lambda}) \prod_{t_{2}=1}^T \emph{P}(\bm{C}^{t_{2}}| \bm{\Lambda}) \prod_{t_{3}=1}^T \emph{P}(\bm{F}^{t_{3}}| \bm{\Lambda}).
\end{flalign*}
Then the tightest lower bound of the observed data log likelihood is
\begin{flalign*}
& \log \emph{P}(\mathcal{D}| \bm{\Lambda}) \geq \sum_{\tau=1}^3 \sum_{t_{\tau}=1}^T \emph{H}(q^{t_{\tau}}) + \emph{E}_{q^{t_{\tau}}} \left( \log \emph{P}(\bm{N}^{t_{\tau}} | \bm{\Lambda}) \right)
\end{flalign*}
where $ q^{t_{1}} = \emph{P}(\bm{N}^{t_{1}} | \bm{R}^{t_{1}}, \bm{\Lambda})$, $q^{t_{2}} = \emph{P}(\bm{N}^{t_{2}} | \bm{C}^{t_{2}}, \bm{\Lambda})$, and $ q^{t_{3}}(\bm{N}^{t_{3}}) = \emph{P}(\bm{N}^{t_{3}} | \bm{F}^{t_{3}}, \bm{\Lambda}) $ are multinomial distributions.
\end{proposition}

In the EM algorithm, the expectation in the E-step is taken with respect to the distribution estimated using the previous iteration's estimate of the parameter $\hat{\bm{\Lambda}}^{k}$, and the M-step does not depend on the entropy terms in the lower bound in Proposition \ref{prop:indep_time_LB}, which are constant with respect to $\bm{\Lambda}$. Since the likelihoods are all Poisson, the E-step reduces to computing the means of multinomial distributions and the M-step for any $ij$ pair is given by the Poisson MLE with the unknown $N_{ij}^t$ terms replaced by their mean values. Explicitly the M-step objective is
\begin{flalign} \label{mle_obj}
& \hat{\Lambda}_{ij}^{k+1} = \underset{\Lambda_{ij}}{\arg\max} \, -\Lambda_{ij} + \log( \Lambda_{ij} ) N_{ij}^{total} 
\end{flalign}
where $ N_{ij}^{total} = \sum_{t_{1}=1}^T \E(N_{ij}^{t_{1}} | \bm{R}^{t_{1}}, \hat{\bm{\Lambda}}^{k}) \\+ \sum_{t_{2}=1}^T \E(N_{ij}^{t_{2}} | \bm{C}^{t_{2}}, \hat{\bm{\Lambda}}^{k}) + \sum_{t_{3}=1}^T \E(N_{ij}^{t_{3}}| \bm{F}^{t_{3}}, \hat{\bm{\Lambda}}^{k})$ and the expectations are with respect to the multinomial distributions of Proposition \ref{prop:indep_time_LB} . Thus the Poisson MLE equals $ \hat{\Lambda}_{ij}^{k+1} = N_{ij}^{total} / 3T $.

\subsubsection{Maximum a Posteriori by EM}

Because there are $P^2$ unobserved variables and only $\mathcal{O}(P)$ observed variables, the expected log likelihoods have a lot of local maxima. In order to make the EM objective better defined and incorporate the baseline Poisson rate information $ \bm{\Lambda}_0 $, a prior can be added to the likelihood model of the previous subsection. The EM objective of this new model is now the expected log posterior and the estimator in the M-step is the maximum a posteriori (MAP) estimator. It is natural to choose a conjugate prior of the form $ \P(\bm{\Lambda} ) = \prod_{ij} \P(\Lambda_{ij} )$ where each $ \Lambda_{ij} \sim Gamma(\epsilon_{ij} \Lambda_{0\, ij} + 1, \epsilon_{ij} ) $ (shape, rate) as this choice yields a closed form expression for the posterior distribution. These priors have modes at the baseline rates $\Lambda_{0 \, ij}$. The hyperparameters $\epsilon_{ij}$ can be thought of as the belief we have in the correctness of the baseline so as $ \epsilon \rightarrow 0 $, the prior variance goes to infinity, and the prior becomes non-informative because we have no confidence in the baseline, while as $ \epsilon \rightarrow \infty $, the prior variance goes to zero, and the prior degenerates into the point $ \Lambda_{0 \, ij}$ because we are certain the baseline is correct. 

Given a matrix of hyperparameters $\bm{\epsilon}$, the complete data posterior distribution is $ \P(\bm{\Lambda} | \bm{\epsilon}, \bm{N}^1, \dots, \bm{N}^T) = \prod_{ij} \P(\Lambda_{ij} | \epsilon_{ij} , N_{ij} ^1, \dots, N_{ij} ^T)$ where each posterior is of the form of $Gamma(\epsilon_{ij} \Lambda_{0 \, ij} + 1 + \sum_{t=1}^T N_{ij}, \epsilon_{ij} + T) $. Because we can only observe the network indirectly $\mathcal{D} = \{\bm{R}^t, \bm{C}^t, \bm{F}^t \}_{t=1}^T$, we again must estimate the mode of this posterior using the EM algorithm, which is very similar to the algorithm for the likelihood model. The only difference is the M-step in which an additional term of the form $ \sum_{ij} (\epsilon_{ij} \Lambda_{0 \, ij}) \log(\Lambda_{ij}) - \epsilon_{ij} \Lambda_{ij} $ is added to \eqref{mle_obj}. Thus at every EM iteration, the entries of the MAP estimator matrix $\hat{\bm{\Lambda}}^{k+1} $ are
\begin{flalign} \label{lambda_map}
\hat{\Lambda}^{k+1}_{ij} = \frac{\epsilon_{ij} \Lambda_{0 \, ij} + N_{ij}^{total} }{\epsilon_{ij} + 3T} 
\end{flalign}
where $ N_{ij}^{total} $ is the same as in \eqref{mle_obj}.

\subsubsection{Bayesian Hierarchical Model}

Choosing the hyperparameters $ \epsilon_{ij} $ can be difficult because it is not always possible to quantify our belief in the correctness of the baseline rates. We can rectify this by allowing the $\epsilon_{ij}$ to be random with hyperpriors $ \epsilon_{ij} \sim Uniform(0, \infty)$. We choose uninformative hyperpriors for $\epsilon_{ij} > 0$. A notional diagram for the proposed hierarchical model is shown in Fig. \ref{fig:gen_model}. 

\begin{figure}[h] 
	\centering
	\includegraphics[width=\linewidth]{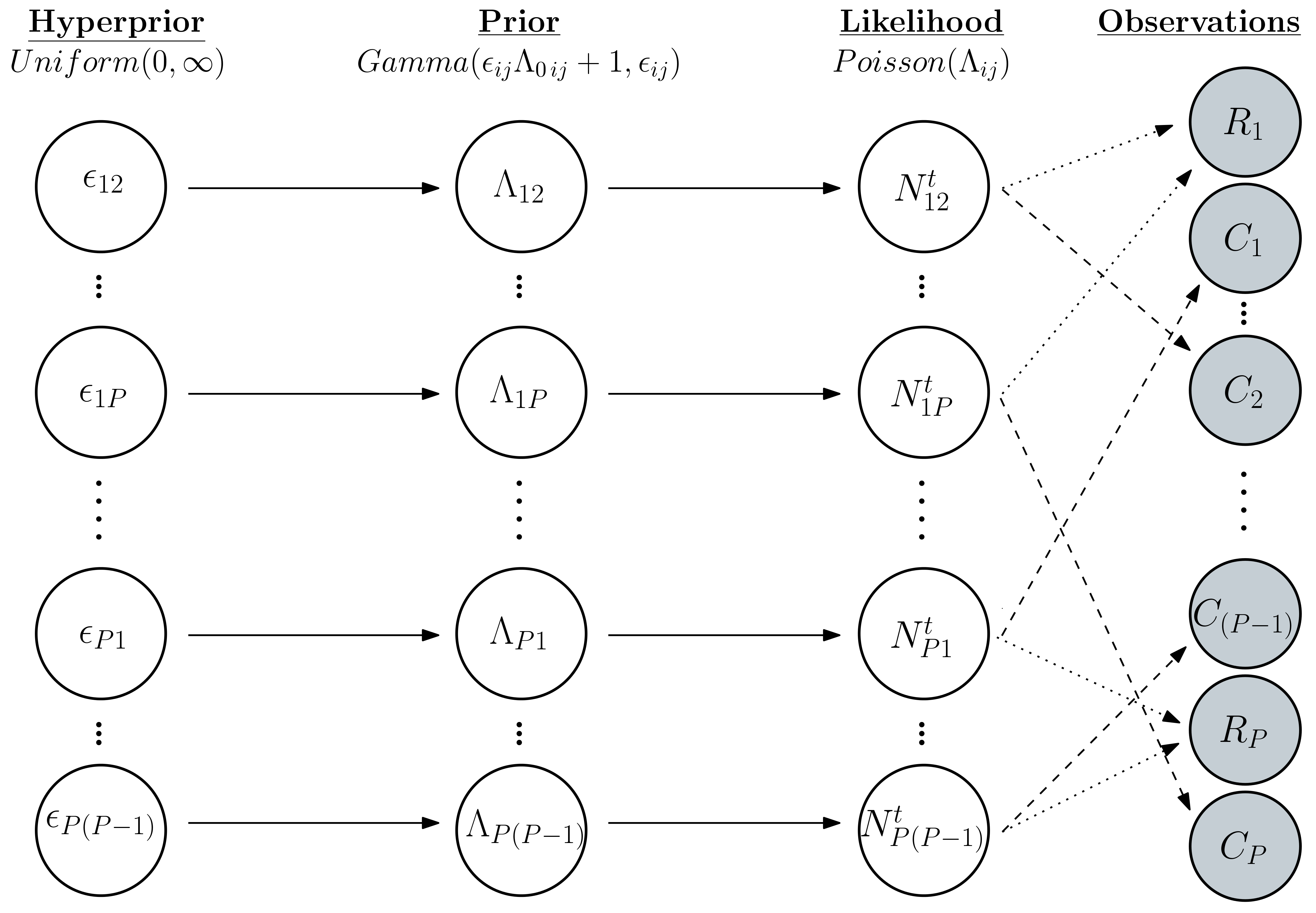}
	\caption{The statistical process believed to underlie our network.}
	\label{fig:gen_model}
\end{figure}

With these uninformative priors the posterior takes the form
\begin{flalign} \label{full_hi_low_bound}
& \P(\bm{\Lambda} | \bm{N}^1, \dots, \bm{N}^T) = \int \frac{\P(\bm{N}^1, \dots, \bm{N}^T | \bm{\Lambda} ) \P(\bm{\Lambda} | \bm{\epsilon}) \P(\bm{\epsilon}) }{ \P(\bm{N}^1, \dots, \bm{N}^T )} \, d\bm{\epsilon} \notag \\
& =\int \frac{\P(\bm{N}^1, \dots, \bm{N}^T | \bm{\Lambda} ) \P(\bm{\Lambda} | \bm{\epsilon}) }{\P(\bm{N}^1, \dots, \bm{N}^T| \epsilon)} \frac{\P(\bm{N}^1, \dots, \bm{N}^T | \bm{\epsilon}) \P(\bm{\epsilon}) }{\P(\bm{N}^1, \dots, \bm{N}^T)} \, d\bm{\epsilon} \notag \\
& = \int \P(\bm{\Lambda} | \bm{\epsilon} , \bm{N}^1, \dots, \bm{N}^T ) \P(\bm{\epsilon} | \bm{N}^1, \dots, \bm{N}^T ) \, d\bm{\epsilon}
\end{flalign}
where $ \P(\bm{\epsilon} | \bm{N}^1, \dots, \bm{N}^T ) = \int \P(\bm{\Lambda}, \bm{\epsilon} | \bm{N}^1, \dots, \bm{N}^T ) \, d\bm{\Lambda}$. The observed (incomplete data) log posterior $ \log\P(\bm{\Lambda} | \mathcal{D}) $ has lower bound proportional to
\begin{flalign*}
& \log \left( \int \exp \bigg\{  \E_{\q} \left(\log \P(\bm{\Lambda} | \bm{\epsilon} , \bm{N}^1, \dots, \bm{N}^T) \right) \hspace{-4pt} \bigg\} \right. \\
& \hspace{33pt} \left. \exp \left\{  \E_{\q} \left( \log \int \P(\bm{\Lambda}, \bm{\epsilon} | \bm{N}^1, \dots, \bm{N}^T) \, d\bm{\Lambda} \right) \right\} d\bm{\epsilon} \right) 
\end{flalign*}
which is tight when $ \q = \P(\bm{N}^1, \dots, \bm{N}^T | \mathcal{D}, \bm{\Lambda})$, as shown in \eqref{lower_bound} in the Appendix.

However, marginalizing the joint posterior $ \int \P(\bm{\Lambda}, \bm{\epsilon} | \bm{N}^1, \dots, \bm{N}^T) \, d\bm{\Lambda} $ is often not feasible, so instead it is popular to use empirical Bayes to approximate it with a point-estimate 

We propose an empirical Bayes approach to maximizing the log posterior as an alternative to maximization of \eqref{full_hi_low_bound} $ \hat{\bm{\epsilon}} = \underset{\bm{\epsilon}}{\arg\max} \, \P(\bm{\epsilon} | \bm{N}^1, \dots, \bm{N}^T ) $. This empirical Bayes approximation can be embedded in the EM algorithm so that once we have an estimate for $\bm{\epsilon}$, an estimator for $ \bm{\Lambda} $ is obtained by maximizing the expected log conditional posterior $  \E_{\q} \left(\log \P(\bm{\Lambda} | \hat{\bm{\epsilon}} , \bm{N}^1, \dots, \bm{N}^T ) \right) $.

\begin{theorem} \label{thm:hbayes}
Using the time independence in Proposition \ref{prop:indep_time_LB} and the empirical Bayes approximation, the E-step of the EM algorithm for the hierarchal model is
\begin{flalign*}
& \hat{N}_{ij}^{t_{1}} = \emph{E}(N_{ij}^{t_{1}} | \bm{R}^{t_{1}}, \hat{\bm{\Lambda}}^{k}) = \frac{\hat{\Lambda}_{ij}^{k}}{\sum_{j=1}^P \hat{\Lambda}_{ij}^{k} } R^{t_{1}}_{i} \\
& \hat{N}_{ij}^{t_{2}} = \emph{E}(N_{ij}^{t_{2}}| \bm{C}^{t_{2}}, \hat{\bm{\Lambda}}^{k}) = \frac{\hat{\Lambda}_{ij}^{k}}{\sum_{i=1}^P \hat{\Lambda}_{ij}^{k} } C^{t_{2}}_{j}, \\
& \hat{N}_{ij}^{t_{3}} = \emph{E}(N_{ij}^{t_{3}} | \bm{F}^{t_{3}}, \hat{\bm{\Lambda}}^{k}) = \frac{\hat{\Lambda}_{ij}^{k}}{\sum_{ij} \hat{\Lambda}_{ij}^{k} } F^{t_{3}}_{h} \text{ for any pair $ij$,}\end{flalign*}
and the M-step is
\begin{flalign*}
& \hat{\epsilon}^{\, k+1}_{ij} = \underset{\epsilon_{ij}}{\arg\max} \sum_{\tau=1}^3 \sum_{t_{\tau}=1}^T \log \frac{\Gamma(\hat{N}_{ij}^{t_{\tau}} + \epsilon_{ij} \Lambda_{0 \, ij} + 1)}{\Gamma(\epsilon_{ij}\Lambda_{0 \, ij} + 1) } \\
& \hspace{12pt} + \hspace{-1pt} \sum_{\tau=1}^3 \sum_{t_{\tau}=1}^T \hspace{-1pt} (\epsilon_{ij} \Lambda_{0 \, ij} + 1) \log \frac{\epsilon_{ij}}{1+\epsilon_{ij}} - \hat{N}_{ij}^{t_{\tau}} \hspace{-1pt} \log(1+\epsilon_{ij}) \\
& \text{ and } \\
& \hat{\Lambda}^{k+1}_{ij} = \underset{\Lambda_{ij}}{\arg\max} \, (\hat{\epsilon}_{ij} \Lambda_{0\, ij} ) \log(\Lambda_{ij}) - \hat{\epsilon}_{ij} \Lambda_{ij} -3T \Lambda_{ij} \\
& \hspace{24pt} + \log(\Lambda_{ij}) \left( \sum_{t_{1}=1}^T \hat{N}_{ij}^{t_{1}} + \sum_{t_{2}=1}^T \hat{N}_{ij}^{t_{2}} + \sum_{t_{3}=1}^T \hat{N}_{ij}^{t_{3}} \right) .
\end{flalign*}

\end{theorem}

Since the function that lower bounds the observed log likelihood changes after every iteration of the EM algorithm, the prior should also change after every iteration. Intuitively, the earlier iterations of the EM algorithm will have expected log likelihoods that are more misspecified than the later iterations. This suggests spreading the prior distribution in the earlier iterations. The empirical Bayes approximation of Theorem \ref{thm:hbayes} effectively does this by allowing the variance of the prior to be chosen using the data instead of fixing it as a constant. In this manner, the empirical Bayes approximation can be thought of as a Bayesian analog to the regularized EM algorithm of \cite{yi2015regularized}.

\subsection{Warm Starting with Minimum Relative Entropy}

The EM algorithm is well known to be sensitive to initialization, especially if the objective has a lot of local maxima. Thus if instead of a random initialization, the EM algorithm is warm-started, it is more likely to converge to a good maximum and also potentially converge faster. A good choice for an initialization point is a more robust estimator of the rate matrix such as the solution to a model with fewer distributional assumptions. Thus instead of modeling an explicit generative model, we can instead adopt the minimum relative entropy (MRE) principle \cite{kullback1997information, Cover:2006, altun2006unifying}, and \cite{koyejo2013representation}. Geometrically, this reduces to an information projection of the prior distribution, as shown in Fig. \ref{fig:MRE}. 

\begin{figure}[h] 
	\centering 
	\includegraphics[width=\linewidth]{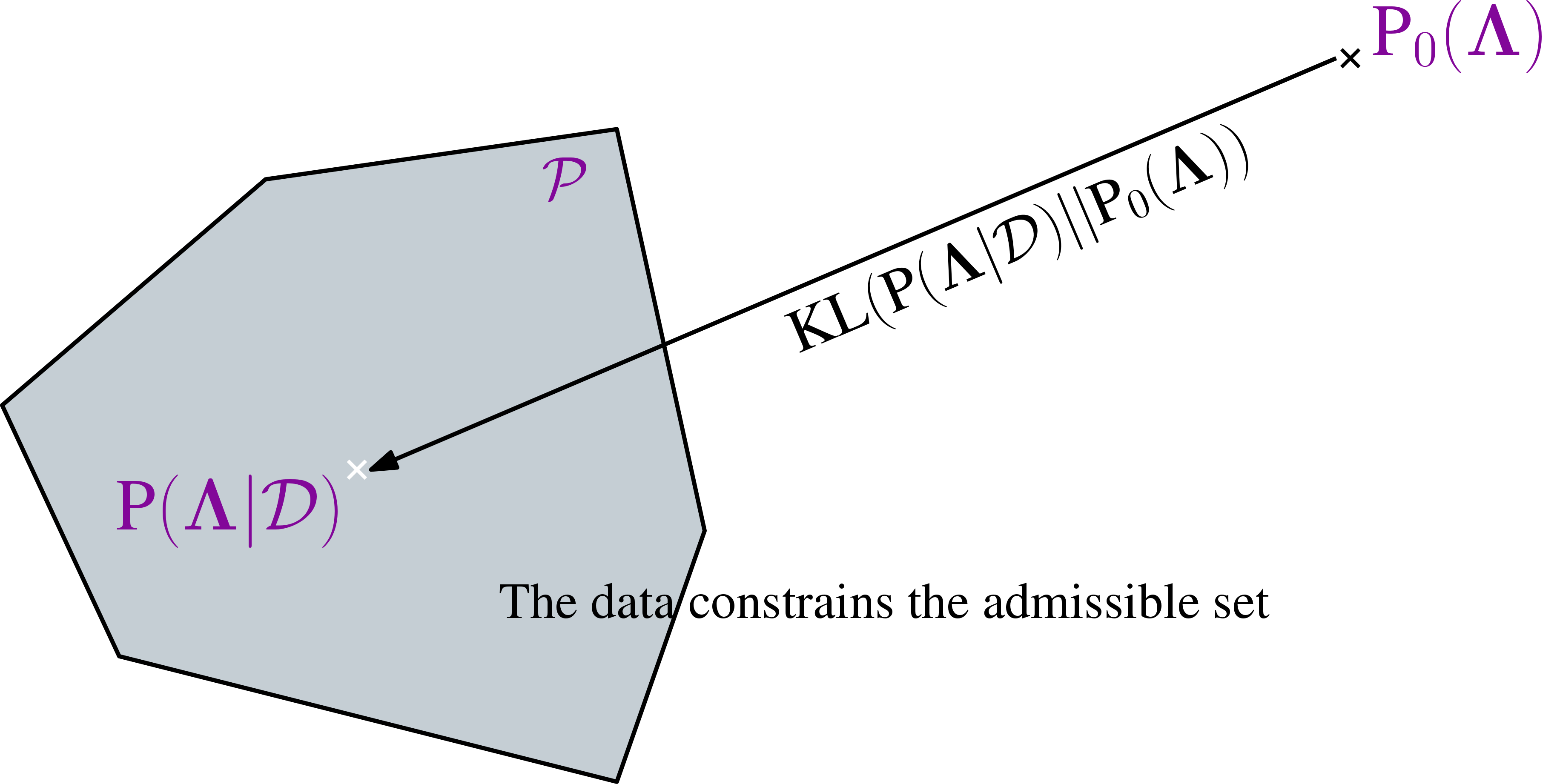}
	\caption{A projection of the prior, $\P_0(\bm{\Lambda})$, onto a feasible set $\mathcal{P}$ of distributions that satisfy the observed data, $\mathcal{D}$.}
	\label{fig:MRE}
\end{figure}

The constrained minimum relative entropy distribution is the density that is closest to a given prior distribution and lies in a feasible set, $\mathcal{P}$. This feasible set is formed from constraints that require their expected values, with respect to the minimum relative entropy distribution, to match properties of the observations, $\mathcal{D}$ (the total ingress, egress, and flows). And because relative entropy is the Kullback-Leibler (KL) divergence between probability distributions, this is used as the metric for closeness. This closeness criterion is well suited to the anomaly detection problem of interest to us because anomalous activity is rare, so the distribution of the actual rates, $\bm{\Lambda}$, should be similar to the prior distribution $\P_0(\bm{\Lambda}) = \P(\bm{\Lambda} | \bm{\Lambda}_0 )$, which is parameterized by the baselines rates $\bm{\Lambda_0}$.

The MRE objective is
\begin{flalign*}
& \underset{\P(\bm{\Lambda} | \bm{R}, \bm{C}, \bm{F}) }{\min} \text{KL} \left(\P(\bm{\Lambda} | \bm{R}, \bm{C}, \bm{F}) || \P_0(\bm{\Lambda} ) \right) \\
& \mathrel{\makebox[.6\linewidth]{\text{subject to} }} \\
& \int \P(\bm{\Lambda} | \bm{R}, \bm{C}, \bm{F}) (\bm{\Lambda} \bm{1} - \bar{\bm{R}}) \, d\bm{\Lambda} = \bm{0} \\
& \int \P(\bm{\Lambda} | \bm{R}, \bm{C}, \bm{F}) ( \bm{1}'\bm{\Lambda} - \bar{\bm{C}} ) \, d\bm{\Lambda} = \bm{0} \\
& \int \P(\bm{\Lambda} | \bm{R}, \bm{C}, \bm{F}) ( \bm{A} \bm{\Lambda} \bm{B} - \bar{\bm{F}}) \, d\bm{\Lambda} = \bm{0} \\
\end{flalign*}
where $\bm{0}$ and $\bm{1}$ are vectors of zeros and ones respectively, $ \bar{\bm{C}} = \frac{1}{T} \sum_{t=1}^T \bm{C}^t $ and $ \bar{\bm{R}} = \frac{1}{T} \sum_{t=1}^T \bm{R}^t $ are the average rates of observed total traffic into and out of each node, and $ \bm{A} $ and $ \bm{B} $ are 0-1 matrices summing the rates that flow through each of the interior nodes with average observations $ \bar{\bm{F}} = \frac{1}{T} \sum_{t=1}^T \bm{F}^t $. Using the Legendre transform of the Lagrangian to get the Hamiltonian, the optimal density has the form
\begin{flalign} \label{re-sol}
& \P(\bm{\Lambda} | \bm{R}, \bm{C}, \bm{F}) = \frac{\P_0(\bm{\Lambda} ) }{Z(\bm{\rho}, \bm{\gamma}, \bm{\phi})} \exp \left\{ \bm{\rho}' (\bm{\Lambda} \bm{1} - \bar{\bm{R}}) \right. \\
& \hspace{90pt} \left. + \bm{\gamma}' (\bm{1}' \bm{\Lambda} -\bar{\bm{C}} ) + \bm{\phi}' (\bm{A} \bm{\Lambda} \bm{B} - \bar{\bm{F}}) \right\} \notag 
\end{flalign}
where $ \bm{\rho}, \bm{\gamma}, \bm{\phi} $ are Lagrange multipliers that maximize the negative log partition function $-\log(Z\left(\bm{\rho}, \bm{\gamma}, \bm{\phi}) \right)$.

\begin{proposition} \label{thm:mre}
Let $ \emph{\P}_0(\bm{\Lambda} ) = \prod_{i j} \emph{\P}_0(\Lambda_{ij}) $ be independent Laplace distributions with mean parameter $ \Lambda_{0 \, ij} $ and scale parameter $1$, then the constrained mode of the MRE distribution is the solution to
\begin{flalign*}
& \underset{\bm{\Lambda} \in \mathbb{R}^{+}}{\arg\max} \, -||\bm{\Lambda} - \bm{\Lambda}_0||_1 + \hat{\bm{\rho}}' (\bm{\Lambda} \bm{1} - \bar{\bm{R}}) \\
& \hspace{97pt} + \hat{\bm{\gamma}}'(\bm{1}' \bm{\Lambda} - \bar{\bm{C}})' + \hat{\bm{\phi}}' (\bm{A} \bm{\Lambda} \bm{B} - \bar{\bm{F}}) 
\end{flalign*}
where $\hat{\bm{\rho}}, \hat{\bm{\gamma}}, \hat{\bm{\phi}} = \underset{\bm{\rho}, \bm{\gamma}, \bm{\phi}}{\arg\max} -\log \left(Z(\bm{\rho}, \bm{\gamma}, \bm{\phi}) \right) $. 
\end{proposition}
Maximizing the above expression over $\bm{\Lambda}$ (constrained to only positive real numbers) can be seen as a slight relaxation of the more direct objective of minimizing the loss function
\begin{flalign} \label{l1-opt}
& \underset{\bm{\Lambda} \in \mathbb{R}^{+}}{\arg\min} \, || \bm{\Lambda} - \bm{\Lambda}_0 ||_1 \\
& \text{subject to} \,\,\bm{\Lambda} \bm{1} = \bar{\bm{R}}, \,\, \bm{1}' \bm{\Lambda} = \bar{\bm{C}}, \,\, \bm{A} \bm{\Lambda} \bm{B} = \bar{\bm{F}} \notag
\end{flalign}
where $|| \cdot ||_1$ is the element wise $\ell_1$ norm. The loss function in \eqref{l1-opt} has the advantage that it can be easily implemented in any constrained convex solver such as CVX \cite{cvx}.

The objective in \eqref{l1-opt} is an easily interpretable formulation for estimating the rate matrix, which does not depend on the unobserved traffic $N^t_{ij}$. And, because it does not put distributional assumptions on the ``likelihood", it is more robust to model mismatch, at the cost of accuracy. The generality of the solution to \eqref{l1-opt}, while not precise enough on its own, makes it a good candidate to be further refined by the EM algorithm in the Hierarchical Poisson model.

\section{Testing For Anomalies} \label{testing}

Since the estimators in the previous section are maximizers of probabilistic models, a natural way to test for anomalies in the rate matrix $\bm{\Lambda}$ is to compare goodness of fit of the fitted model using hypothesis testing. By testing the null hypothesis $ \v(\bm{\Lambda}) = \v(\bm{\Lambda}_{0} ) $ against the alternative hypothesis $ \v(\bm{\Lambda}) \neq \v(\bm{\Lambda}_{0} ) $, we can control the false positive rate (FPR) (Type 1 error), of incorrectly declaring anomalous activity in the rate matrix, using a level-$\alpha$ test. In this section we will represent a statistical model with the notation $\mathcal{M}(\cdot)$, as the results apply for both log likelihood and log posterior models.

Depending on if the statistical models are likelihoods or posteriors, the statistic 
\begin{flalign} \label{correct_stat}
\psi = -2 \sum_{t=1}^T \left( \log(\mathcal{M}_t(\bm{\Lambda}_0)) - \log(\mathcal{M}_t(\hat{\bm{\Lambda}})) \right)
\end{flalign}
would be either a log likelihood ratio (LR) statistic or a log posterior density ratio (PDR) statistic \cite{basu1996bayesian} respectively, where $\hat{\bm{\Lambda}} = \underset{\bm{\Lambda} \in \mathbb{R}^+ }{\arg\max} \, \mathcal{M}(\bm{\Lambda}) $. Thus testing $\psi$ against a threshold can be seen as a generalized log likelihood ratio test or generalized log posterior ratio test with a composite alternative hypothesis. 

\begin{proposition} \label{prop:pdr}
Under the standard regularity conditions for the log LR statistic or under the sufficient conditions of the Bernstein-von Mises theorem for the log PDR statistic, $\psi$ will be asymptotically $\chi^2_{P^2 - P}$ distributed under the null hypothesis.
\end{proposition}

Next we show that the statistic $\psi$ in \eqref{correct_stat} is a good estimator of the KL divergence between the true model at its maximum and the true model at the baseline. And even if the models are misspecified, the statistic
\begin{flalign*}
\hat{\psi} = -2 \sum_{t=1}^T \left( \log( \hat{\mathcal{M}}_t^k(\bm{\Lambda}_0) ) - \log(\hat{\mathcal{M}}_t^k(\hat{\bm{\Lambda}})) \right) 
\end{flalign*}
can still be a good estimator for goodness-of-fit, where the $k$ in $\hat{\mathcal{M}}^k(\bm{\Lambda}_0)$ and $ \hat{\mathcal{M}}^k(\hat{\bm{\Lambda}})$ indicates the iteration of the EM algorithm. 

\begin{proposition} \label{prop:miss}
The statistic $\psi/T$ is a consistent estimator for 
\begin{flalign*}
\Psi = 2 \, \emph{KL}\left(\mathcal{M}(\bm{\Lambda}^*)|| \mathcal{M}(\bm{\Lambda}_0) \right),
\end{flalign*}
 the KL divergence between the true model and the true model under the null hypothesis. The statistic $\hat{\psi}/T$ is a consistent estimator for 
\begin{flalign} \label{miss_err}
& 2 \, \emph{KL}\left(\mathcal{M}(\bm{\Lambda}^*)|| \mathcal{M}(\bm{\Lambda}_0) \right) \\
& \hspace{15pt} -2 \left( \emph{KL} ( \mathcal{M}(\bm{\Lambda}^*) || \hat{\mathcal{M}}^k(\hat{\bm{\Lambda}}^*) ) - \emph{KL} ( \mathcal{M}(\bm{\Lambda}_0) || \hat{\mathcal{M}}^k(\bm{\Lambda}_0) ) \hspace{-1pt} \right) \notag
\end{flalign}
where $ \hat{\mathcal{M}}^k(\hat{\bm{\Lambda}}^*) $ is the closest population local maximum at iteration $k$. 
\end{proposition}

The second term in \eqref{miss_err} can be seen as the difference between the true model misspecification error and the model misspecification error of the null hypothesis. So if conditions are satisfied so that the EM algorithm converges to the global maximum as the number of iterations $k \rightarrow \infty$ or if the model is equally as misspecified under the truth as under the null hypothesis such that the differences in the second term in \eqref{miss_err} cancel to 0, then the statistic $\hat{\psi}/T$ is also a consistent estimator of $\Psi$. The justification for using misspecified models can also be geometrically interpreted as follows. Because the models estimated from the EM algorithm are from the correct parametric family of distributions, the misspecified models still lie on the same Riemannian manifold as the correct models. Below, we provide an algorithm for performing hypothesis testing on the statistic $\hat{\psi}$.

\setlength{\extrarowheight}{2pt}
\setlength{\arrayrulewidth}{1pt}
\begin{table}[h]
\centering
\begin{tabular}{l} \hline 
\textbf{Algorithm 1:} Anomaly Test \\ \hline 
\begin{minipage}{.8\linewidth}
\smallskip 
\begin{algorithmic}
	\STATE {\bfseries Input: } models $ \hat{\mathcal{M}}_1^k, \dots, \hat{\mathcal{M}}_T^k$, critical value $c = F^{-1}(\alpha)$ \\ \hspace{24pt} where $F$ is $\chi^2_{p^2-p}$ CDF, $\alpha$ is test level
	\STATE Solve $\hat{\bm{\Lambda}} = \underset{\bm{\Lambda} \in \mathbb{R}^+ }{\arg\max} \, \sum_{t=1}^T \log( \hat{\mathcal{M}}_t^k(\bm{\Lambda})) $
	\STATE $\hat{\psi} = -2 \sum_{t=1}^T \left( \log( \hat{\mathcal{M}}_t^k(\bm{\Lambda}_0) ) - \log(\hat{\mathcal{M}}_t^k(\hat{\bm{\Lambda}})) \right) $
	\IF{$\hat{\psi} > c $}
	\STATE Reject $ \v(\bm{\Lambda}) = \v(\bm{\Lambda}_{0} ) $ 
	\ELSE
	\STATE Do not reject $ \v(\bm{\Lambda}) = \v(\bm{\Lambda}_{0} ) $ 
	\ENDIF
	\STATE {\bfseries Return: } Reject or Not
\end{algorithmic} 
\smallskip
\end{minipage} \\ \hline 
\end{tabular}
\end{table}
Algorithm 1 calculates the statistic $\hat{\psi}$ as a log ratio of the modes of the model under the null and alternative hypothesis. It then tests $\hat{\psi}$ against a critical value $c$, which is related to the false positive level.

Under the null hypothesis, the statistic $\hat{\psi}$ can be decomposed as sampling error $ - 2 \sum_{t=1}^T \log \hat{\mathcal{M}}_t^k(\hat{\bm{\Lambda}}^*) - \underset{\bm{\Lambda} \in \mathbb{R}^+}{\max} \log \hat{\mathcal{M}}_t^k(\bm{\Lambda}) $ plus model error $ -2 \sum_{t=1}^T \log \hat{\mathcal{M}}_t^k(\bm{\Lambda}_0) - \log \hat{\mathcal{M}}_t^k(\hat{\bm{\Lambda}}^*) $.
Thus for the level-$\alpha$ test $\P(\hat{\psi} > c | \mathcal{H}_0) = \alpha$, a Type-I error can occur due to either sampling error or model error or a combination of both. Since typically the finite sample distribution of the statistic $\psi$ is unknown, the asymptotic distribution described in Proposition \ref{prop:pdr} can be used to choose the critical value $c$ of $ \P(\psi > c | \mathcal{H}_0) = \alpha$. Assuming the model error is small, or small relative to the sampling error, we can also use Proposition \ref{prop:pdr} to choose the critical value of a test with a misspecified statistic $\P(\hat{\psi} > c | \mathcal{H}_0) = \alpha$. In the following section, we will show in simulations that the asymptotic distribution of the correct statistic $\psi$ is adequate for choosing the critical value of a test using the misspecified statistic $\hat{\psi} $.

\section{Computational Complexity}

In Algorithm 2, we present our hierarchical Poisson EM model warm started at the MRE estimator and analyze its computational complexity.

\setlength{\extrarowheight}{2pt}
\setlength{\arrayrulewidth}{1pt}
\begin{table}[h]
	\centering
	\begin{tabular}{l} \hline 
		\textbf{Algorithm 2:} HP-MRE \\ \hline 
		\begin{minipage}{.9\linewidth}
			\smallskip 
			\begin{algorithmic}
		\STATE {\bfseries Input: } observations $\mathcal{D} = \{\bm{R}^t, \bm{C}^t, \bm{F}^t \}_{t=1}^T$, test level $\alpha$
		\STATE Initialize: $\hat{\bm{\Lambda}}$ as the solution to \eqref{l1-opt} 
		\REPEAT
		\STATE E-Step: Calculate $\hat{N}_{ij}^{t_{1}}, \hat{N}_{ij}^{t_{2}}, \hat{N}_{ij}^{t_{3}} $ for all $i, j$ in Theorem \ref{thm:hbayes}
		\STATE M-Step: Solve for $\hat{\epsilon}^{\, k+1}_{ij} $ and $ \hat{\Lambda}^{k+1}_{ij} $ for all $i, j$ in Theorem \ref{thm:hbayes}
		\UNTIL{convergence}
		\STATE Test: Calculate $\hat{\psi}$ and reject if it is greater than critical value $c$
		\STATE {\bfseries Return: } Reject or Not
			\end{algorithmic} 
			\smallskip
		\end{minipage} \\ \hline 
	\end{tabular}
\end{table}

Warm starting the EM algorithm at the MRE solution \eqref{l1-opt} requires using interior-point methods, which have polynomial complexity in the number of variables. Since the MRE objective has $P^2$ linear variables and $2 P^2$ second order cone problem variables, the computational cost is of order $\mathcal{O}( \# IP iter (3 P^2)^r )$ where $r$ is the polynomial degree (often 3) and $\# IP iter $ is the number of iterations of the interior point algorithm.

The E-Step consists of calculating the multinomial means using the observed data. Assume that the number of flows in the interior nodes are roughly $P$, so that each of the row sums, column sums, and interior node flows are the summation of $P$ values. Then for each independent time instance $t_{\tau}$, there are $P$ summations of $P$ values in denominator and a multiplication and division operation on each of the $P^2$ entries in the numerator. The total computational cost of the E-step is of order $\mathcal{O}( \tau T P^2)$ where $\tau$ is the number of different time points in Proposition \ref{prop:indep_time_LB} (2 + number of interior nodes). 

In the M-step, the estimator $ \hat{\epsilon}^{\, k+1}_{ij} $ can only be solved numerically because the score function of the negative binomial distribution is a non-linear equation. Because we can derive the gradient of the score function, we can use a trust-region method with a Newton conjugate gradient subproblem (each subproblem has linear complexity in time points). Given $\hat{\epsilon}^{\, k+1}_{ij} $, the estimator $ \hat{\Lambda}^{k+1}_{ij} $ can be solved in closed form \eqref{lambda_map} with scalar operations, making its complexity linear in time points. Thus the total computational cost of the M-step is of order $\mathcal{O}( (1+\# CG iter) T P^2)$ where $\# CG iter$ is the number of conjugate gradient iterations.

Given the final iterations EM estimators, evaluating the models at each $i, j$ entry only involves scalar operations, and getting the log ratio statistic $\hat{\psi} $ requires summing over all $i, j$ entries and the $T$ time points; so the total complexity of the anomaly test statistics is of order $\mathcal{O}( T P^2)$. Thus, overall Algorithm 2 has computational complexity of order $\mathcal{O}( \# IP iter (3 P^2)^r + \# EM iter ( (\tau+1) T P^2 + \# CG iter T P^2) )$. Note that our choice in algorithms for the numerical optimizations were based more on convenience (using popular standard packages e.g. CVX, Matlab's fsolve) than optimal performance, so the computational complexities listed in this section are certainly not the best case scenarios. Nonetheless, even using non-optimal numerical algorithms, we show, in the following section, that our method can run in a reasonable amount of time in both simulations and large real world problems.

\section{Simulation and Data Examples}

In this section, we model network traffic in both simulated and real datasets as hierarchical Poisson posteriors to get estimators of the true network traffic rates. These estimators, from the hierarchical Poisson posteriors where the EM algorithm is initialized randomly or at the MRE estimator (Rand-HP or MRE-HP), are tested against baseline rates to detect anomalous activity in the network, as shown in Algorithm 1. We compare the performance of our proposed models to the maximum likelihood EM (MLEM) model of \cite{Vanderbei94anem} (with the same time independence assumptions of Proposition \ref{prop:indep_time_LB} for feasibility), the Traffic and Anomaly Map (TA-Map) method of \cite{7098434}, and an ``Oracle" that unrealistically observes the network directly. The ``Oracle" estimator is the uniformly minimum variance unbiased estimator and achieves the Cramer-Rao lower bound \cite{casella2002statistical}.

The Traffic and Anomaly Map method is the state-of-the-art for estimating the rates in networks with traffic anomalies. Specifically for the TA-Map method we use the objective of (P1) in \cite{7098434}, but with the low rank decomposition of (P4) in \cite{7098434} where $\bm{X}=\bm{L} \bm{Q}'$ and $\bm{Q} = \bm{1} $ is a vector of ones because the rates do not change over time. Since the anomalies also do not change over time, they can be expanded as $\bm{A} \bm{Q}'$ where $\bm{A}$ is a $P^2 \times 1$ vector of rates of anomalous activity. We use $\bm{\Lambda}_0$ to form the routing matrix for the vector of nominal rates $\bm{L}$ and a full routing matrix for the vector of anomalous rates $\bm{A}$ since we do not know any structural knowledge about them. Additionally, converting the notation of \cite{7098434} to the notation of this paper, $\bm{Y} = [\bm{C}, \bm{R}, \bm{F}]$, $\bm{Z}_\Pi$ are defined as the edges that are observed, and $\bm{L}+\bm{A} = \v(\bm{\Lambda})$, where $\bm {L}$ and $\bm{A}$ are solved using CVX on (P1) in \cite{7098434}. We empirically choose the penalty parameters $\lambda_\star = 0.5$ and $\lambda_1 = 0.1$.

\subsection{Simulation Results}

We simulate networks where the baseline rate matrix has 10 exterior nodes and 2 interior nodes. The probability of an edge between any two nodes in the baseline network is $0.65$, the baseline rates $\Lambda_{0 \, ij}$ are drawn from $Gamma(1.75, 1)$ distributions, and each interior node observes the total flow of a random 7 edges. We consider scenarios where anomalous activity can take place in either the edges or the nodes. In the first scenario, the anomalous activity can cause increases in the rates of some of the edges, new edges to appear or disappear, or both. So, the rates of anomalous activity $\Lambda_{ij} - \Lambda_{0 \, ij}$ are drawn from $Gamma(0.75, 1)$ distributions where the probability of anomalous activity between any two nodes is $0.2$. In the second scenario, there is a hidden node that is interacting with the other nodes, thus affecting the observed total flows of the known nodes. So the entries of the true rate matrix are drawn from $Gamma(1.75, 1)$ distributions, but the true rate matrix has 11 exterior nodes and the baseline rate matrix is the $10 \times 10$ submatrix of known nodes. Like in the first scenario, the probability of an edge between the hidden node and another node is $0.2$. All simulations contain 200 trials, with anomalous activity in approximately half of them. 

In Fig. \ref{fig:acc} we explore the accuracy of correctly identifying anomalous activity as a function of the percentage of observed edges, where we observe $T=100$ time points (samples). We measure accuracy as $ \frac{\#TP + \#TN}{\#Trials} $ where the number of true positives (TP) and true negatives (TN) are the number of times a method correctly detects that there is anomalous activity or no anomalous activity respectively. For the probabilistic models (MLEM, Rand-HP, MRE-HP) , we use the likelihood or posterior density ratio tests described in Section \ref{testing} where the critical value is calculated using the inverse cumulative distribution function of the $\chi^2_{P^2-P} $ distribution at $0.05$. The Traffic and Anomaly Map method uses a threshold on the maximum (absolute) value of the anomaly matrix $\bm{A}$ where the threshold is chosen so that it has $0.05$ Type-I error. While the accuracy of all the probabilistic models increases as the percentage of observed edges increases, the MLEM has low accuracy unless over 80\% of the network is observed whereas the two Hierarchical Poisson models have  high accuracy even when no part of the network is directly observed. The TA-Map method also has poor performance at all percentages of the network observed. This may due to issues the TA-Map method has at separating $\bm {L}$ and $\bm{A}$ into the correct separate matrices even when the total estimator $\bm{L}+\bm{A} $ is accurate.

\begin{figure}[h] 
	\centering
	\includegraphics[width=\linewidth]{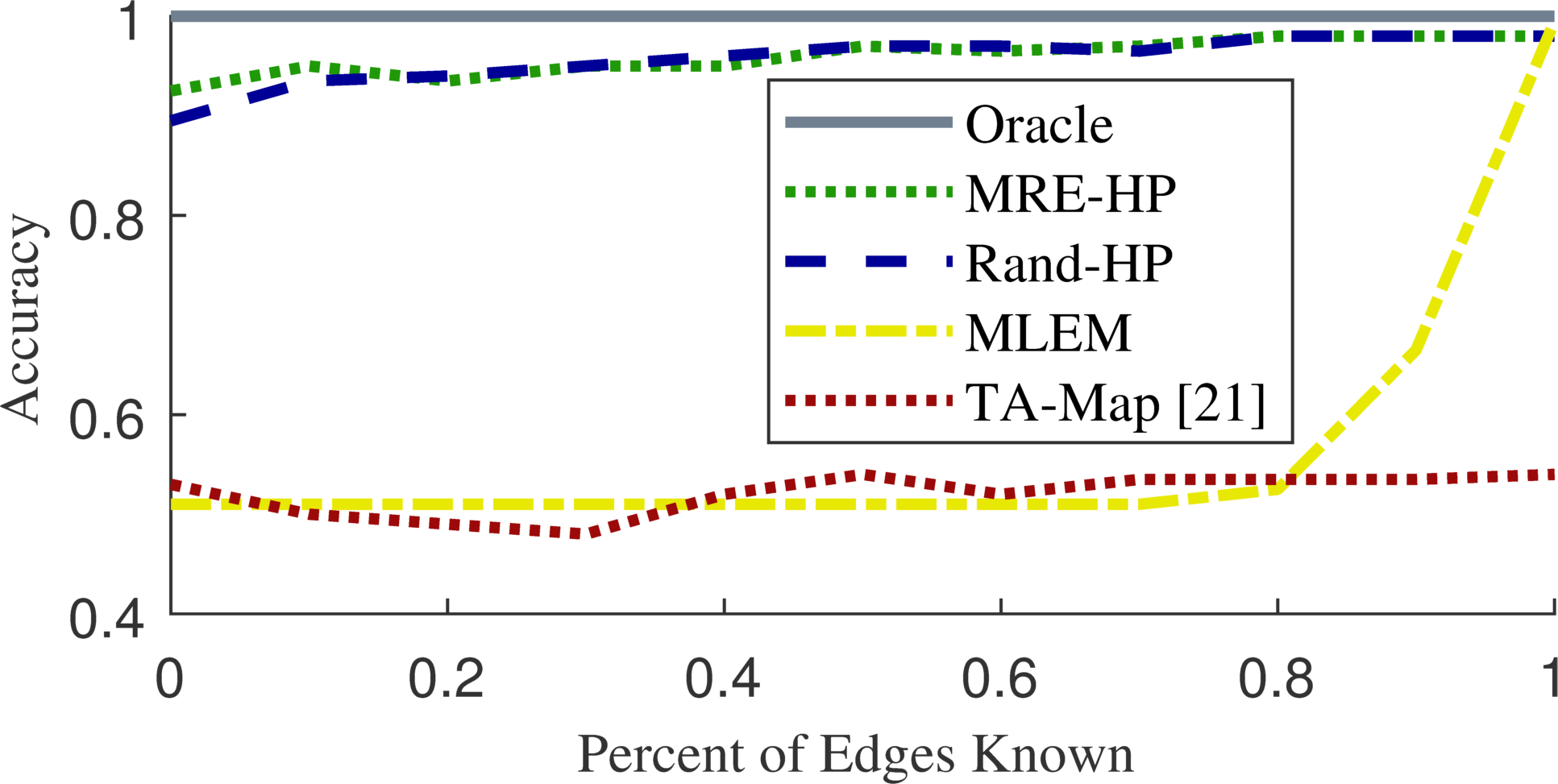}
	\caption{The network has 10 exterior nodes, 2 interior nodes, 35\% sparsity, and a 0.5 probability of having anomalous activity, where $T=100$ samples are observed. The accuracy of correctly detecting if the network has anomalous activity increases as the number of edges observed increases. The proposed Rand-HP, and MRE-HP models outperform the state-of-the-art TA-Map anomaly detector.}
	\label{fig:acc}
\end{figure}

While the Rand-HP and MRE-HP models have approximately the same accuracy at detecting anomalies (MRE-HP does slightly better when only a few of the edges are observed), initializing the EM algorithm of the Hierarchical Poisson model at the MRE solution has additional benefits. Fig. \ref{fig:iter} shows that the EM algorithm in the Hierarchical Poisson model with random initialization takes longer to converge than if it is initialized at the MRE solution. This is because, if the EM algorithm is initialized in a place where likelihood is very noisy, it may have difficultly deciding on the best of the nearby local maxima, but the MRE solution is often already close to a good local maximum. 

Fig. \ref{fig:MSE} shows the mean squared error (MSE) of the estimated rate matrices $ || \hat{\bm{\Lambda}} - \bm{\Lambda} ||_F^2 $. The MRE-HP model gains some of the advantages of the MRE estimator making its MSE much lower than that of the Rand-HP model. As the percentage of observed edges  in the network increases, all estimators' errors decrease to the Oracle estimator's error, which is the lowest possible MSE among all unbiased estimators. However, both the TA-Map method and the MLEM model do not have good performance except when almost all of the network is observed, at which point every estimator performs well. Note that estimating the traffic is not the end goal in the considered anomaly detection problem. We demonstrate this by comparing Fig. \ref{fig:MSE} to Fig. \ref{fig:acc}, where we can see that estimating the traffic well (having low MSE) does not guarantee the method high accuracy. Low MSE implies that a method's estimates do not have a large difference with the true rates, however depending on where the differences occur, it can be enough to cause the method to incorrectly detect anomalous activity. 

\begin{figure}[h] 
	\centering
	\includegraphics[width=\linewidth]{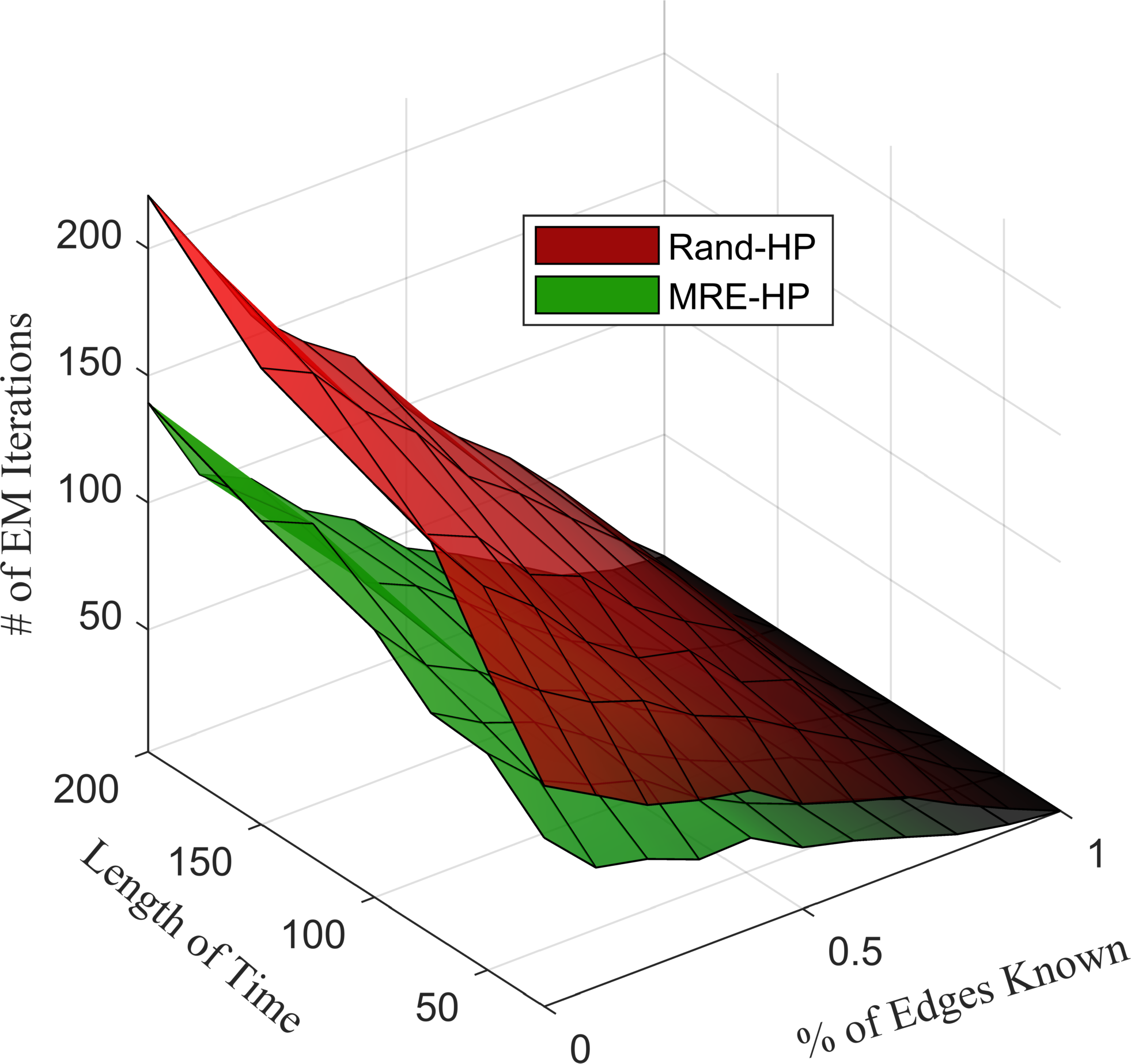}
	\caption{The number of iterations required for the EM algorithm to converge as the observation time and number of edges observed vary. By warm-starting the EM algorithm at the MRE estimator, the number of iteration is much fewer everywhere because it is already close to a good local maximum.}
	\label{fig:iter}
\end{figure}

\begin{figure}[h] 
	\centering
	\includegraphics[width=\linewidth]{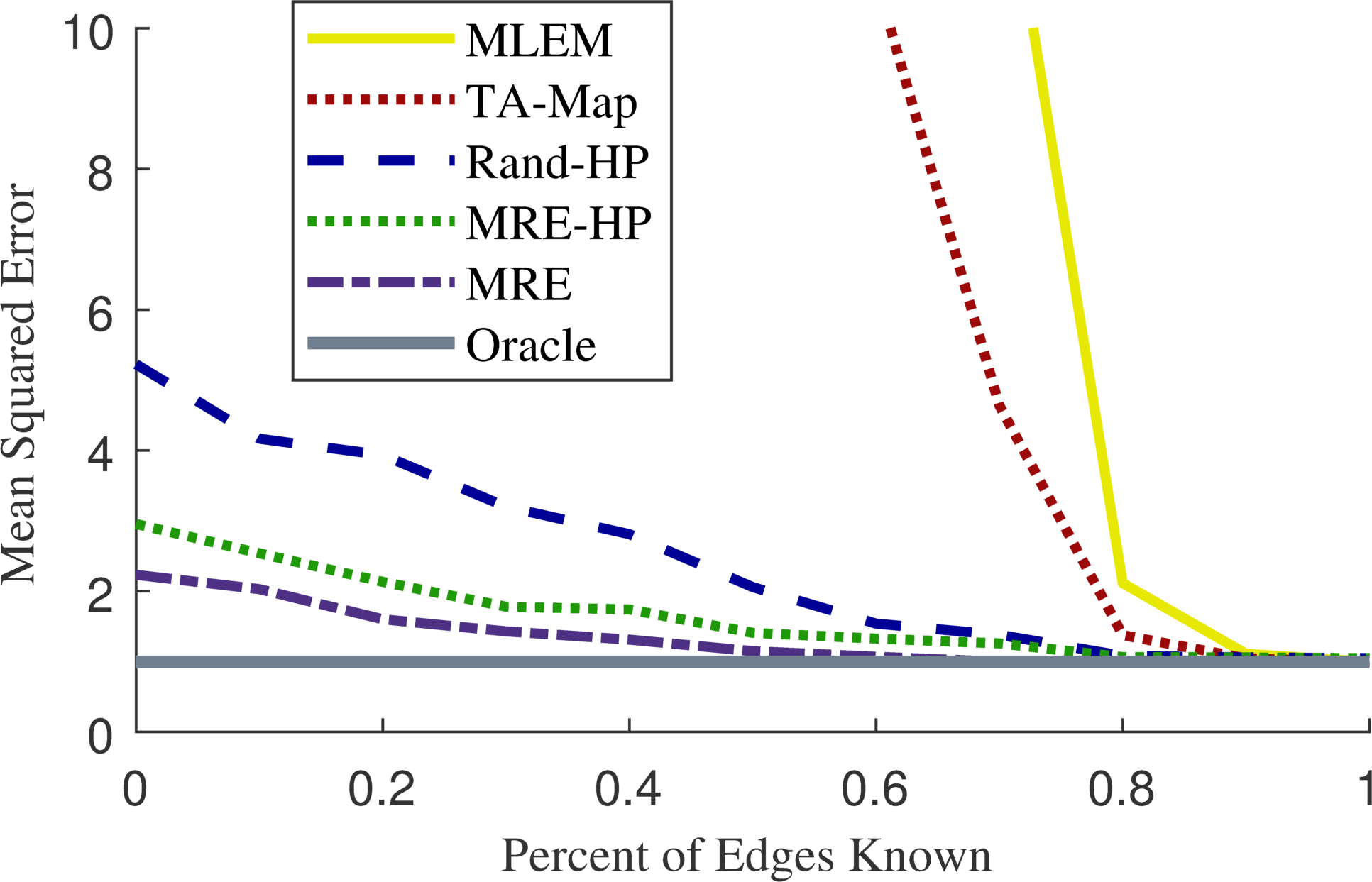}
	\caption{The MSE decreases as the number of edges observed increases. The proposed MRE, Rand-HP, and MRE-HP models outperform the state-of-the-art TA-Map method.}
	\label{fig:MSE}
\end{figure}

Fig. \ref{fig:ROC} shows the ROC curves of the anomaly detection performance of the MRE-HP, MLEM, and TA-Map methods for both the anomalous rates and the hidden node scenarios, where only 20\% of the edges are observed. The accuracy of the MRE-HP model increases with the total observation time $T$, and it can detect anomalous activity almost perfectly with only 100 time points, as evidenced by its area under the curve (AUC) being very close to 1. The stars over the lines are the FPR vs TPR when using the critical values found by calculating the inverse cumulative distribution function of the $\chi^2_{P^2-P} $ distribution at 0.05. The ROC curve for testing a misspecified LR test statistic using the MLEM is just the point at $(1, 1)$ because the Poisson MLE model is so misspecified, it always rejects the null hypothesis. The TA-Map method, while it does not always rejects the null hypothesis like the MLEM model, performs about as bad as random guessing (a diagonal line from $(0, 0)$ to $(1, 1)$). These results are consistent with the accuracy results shown in Fig. \ref{fig:acc}.

\begin{figure}[h]
	\centering
		\includegraphics[width=\linewidth]{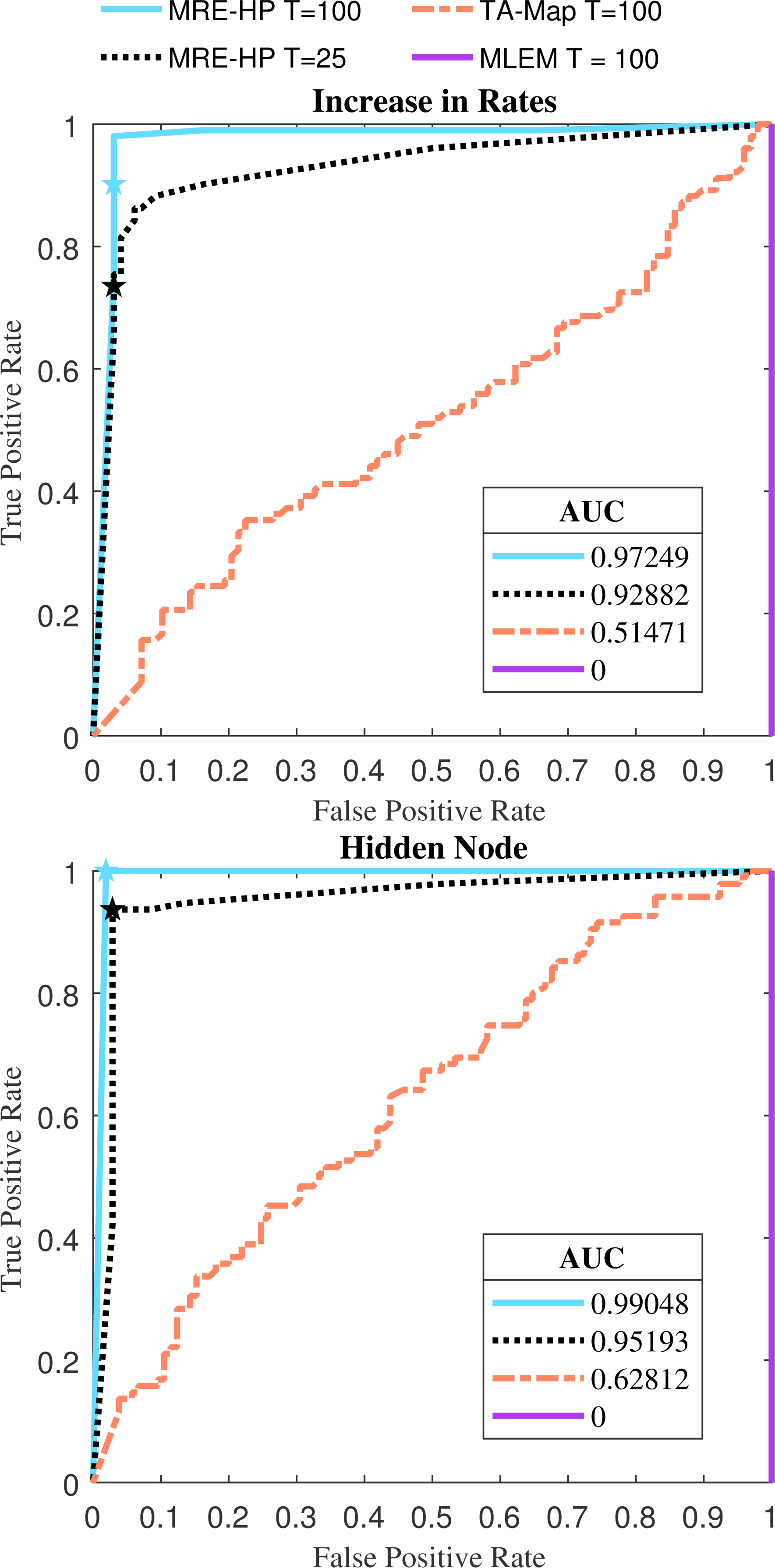}
	\caption{ROC curves where 20\% of edges in the network are observed and roughly half of the networks have anomalous activity. The proposed MRE-HP model can detect anomalous activity almost perfectly while the TA-Map and MLEM methods have poor performance.}
		\label{fig:ROC}
\end{figure}

In Table \ref{Fig7_cpu}, we show the corresponding CPU timings of each method in the two scenarios used in Fig. \ref{fig:ROC}. The algorithms were run on an Intel Xeon E5-2630 processor at 2.30GHz without any explicit parallelization; however some of the built-in Matlab functions are by default multi-threaded (such as ones that call BLAS or LAPACK libraries). While the MRE-HP is slower than the competing methods, its computation time is still very fast and on average less than half a minute. Also, note the significant performance improvement provided by MRE-HP in the considered anomaly detection problem (see Fig. \ref{fig:acc} and Fig. \ref{fig:ROC}).

\setlength\extrarowheight{2pt}
\begin{table}[h]
	\centering
	\caption{Fig. \ref{fig:ROC} CPU Times (in seconds) over 200 Trials}
	\label{Fig7_cpu}
	\begin{tabular}{|c|cc|cc|} \hline
	  & \multicolumn{2}{c|}{Increase in Rates }  & \multicolumn{2}{c|}{Hidden Node} \\ 
		 & Average & Standard Dev. & Average & Standard Dev. \\ \hline
		MRE-HP & 18.594 & 28.611 & 18.901 & 30.785 \\
		MLEM & 0.0398 & 0.0112 & 0.0380 & 0.0119 \\ 
		TA-Map & 3.1860 & 0.1347 & 3.1861 & 0.1912 \\ \hline
	\end{tabular}
\end{table}

\subsection{CTU-13 Dataset}

The proposed model was applied to botnet traffic networks from the CTU-13 dataset, which come from 13 different scenarios of botnets executing malware attacks captured by CTU University, Czech Republic, in 2011 \cite{CTU-13}. The dataset contains real botnet traffic mixed with normal traffic and background traffic and the authors of \cite{CTU-13} processed the captured traffic into bidirectional NetFlows and manually labeled them. Because the objective is to detect if there is botnet traffic among the regular users, we will only use the sub-network of nodes that are being used for normal traffic, but the traffic on this sub-network can be of any type: normal, background, or botnet. Thus, baseline traffic on the network is either normal or background traffic and the anomalous traffic is from botnets. And because the botnet traffic originates and also potentially ceases from nodes that are not the regular users, the anomalous activity is due to unobserved hidden nodes.

The observations consist of the total ingress and egress of each node along with the total flows of 10 interior nodes, where each interior node receives flow from $0.7 P$ other nodes, in addition to observing 20\% of the edges in the network. An observation or sample is all the traffic that occurs in a one-hour time period. For each of the scenarios, we test the probabilistic models at an alpha level of 0.05 under both regimes where the null hypothesis is true (no botnet traffic) and not true (botnet traffic). For the TA-Map method of \cite{7098434}, we use the ROC curves from the simulations to choose the threshold that yields a Type-I error equal to $0.05$. Table \ref{botnet_data} summarizes the characteristics of each of the 13 difference scenarios.

\setlength\extrarowheight{2pt}
\begin{table}[h]
	\centering
	\caption{CTU Network Characteristics}
	\label{botnet_data}
	\begin{tabular}{|c|ccccc|} \hline
	& Time & \# of & \# of Edges & \# of & \# of Edges \\ 
		 	\hspace{-6pt} Scenario \hspace{-6pt}  & $T$ & Nodes & Normal & Hidden & Botnet \\ 
		  & (Hours) & $P$ & Traffic & Nodes & Traffic \\ \hline
		1 & 7 & 510 & 1566 & 2280 & 4428 \\
		2 & 6 & 114 & 249 & 283 & 337 \\ 
		3 & 68 & 333 & 977 & 2463 & 2466 \\ 
		4 & 5 & 414 & 1737 & 9 & 27 \\
		5 & 2 & 246 & 652 & 59 & 67 \\ 
		6 & 3 & 200 & 380 & 2 & 5 \\ 
		7 & 2 & 93 & 161 & 11 & 14 \\ 
		8 & 20 & 3031 & 8799 & 57 & 106 \\ 
		9 & 6 & 485 & 1799 & 706 & 3372 \\ 
		10 & 6 & 260 & 1088 & 25 & 131 \\ 
		11 & 1 & 53 & 162 & 7 & 19 \\
		12 & 2 & 290 & 697 & 861 & 1829 \\
		13 & 17 & 272 & 814 & 267 & 345 \\ \hline
	\end{tabular}
\end{table}

Table \ref{botnet_test} shows that the Hierarchical Poisson model initialized at the MRE solution always correctly rejects the null hypothesis when it is not true. However, the model incorrectly rejects the null hypothesis in Scenario 3. This scenario has far more nodes than any of the other scenarios, and as the number of nodes increase, the number of entries that must be estimated, $\mathcal{O}(P^2)$, vastly outweigh the number of observations, $\mathcal{O}(P)$. This gives rise to a large model misspecification error in this scenario, which would negatively impact the accuracy of Algorithm I. Like in the simulations, the Poisson MLE model always rejects the null hypothesis due to its massive model misspecification error and the TA-Map method also has poor performance in the scenarios that are computationally feasible for the method (the ones marked NA are too computationally expensive). Overall MRE-HP has good performance detecting anomalous activity, especially compared to the other methods.

\setlength\extrarowheight{2pt}
\setlength{\tabcolsep}{5pt}
\begin{table}[h]
	\centering
	\caption{CTU Network Test}
	\label{botnet_test}
	\begin{tabular}{|c|ccc|ccc|} \hline
	\hspace{-6pt} Scenario \hspace{-6pt} & \multicolumn{3}{c|}{When $\mathcal{H}_0$ is True }  & \multicolumn{3}{c|}{When $\mathcal{H}_A$ is True } \\ 
		 & MRE-HP & \hspace{-4pt} MLE \hspace{-4pt} & TA-Map & MRE-HP & \hspace{-4pt} MLE \hspace{-4pt} & TA-Map\\ \hline
		1 & \checkmark & $\times$ & NA & \checkmark & \checkmark & NA \\ 
		2 & \checkmark & $\times$ & $\times$ & \checkmark & \checkmark & \checkmark \\ 
		3 & $\times$ & $\times$ & NA & \checkmark & \checkmark & NA \\ 
		4 & \checkmark & $\times$ & NA & \checkmark & \checkmark & NA \\
		5 & \checkmark & $\times$ & NA & \checkmark & \checkmark & NA \\ 
		6 & \checkmark & $\times$ & NA & \checkmark & \checkmark & NA \\ 
		7 & \checkmark & $\times$ & $\times$ & \checkmark & \checkmark & \checkmark \\ 
		8 & \checkmark & $\times$ & NA & \checkmark & \checkmark & NA \\ 
		9 &\checkmark & $\times$ & NA & \checkmark & \checkmark & NA \\ 
		10 & \checkmark & $\times$ & NA & \checkmark & \checkmark & NA \\ 
		11 & \checkmark & $\times$ & $\times$ & \checkmark & \checkmark & \checkmark \\ 
		12 & \checkmark & $\times$ & NA & \checkmark & \checkmark & NA \\ 
		13 & \checkmark & $\times$ & NA & \checkmark & \checkmark & NA \\ \hline
	\end{tabular}
\end{table}

In Table \ref{botnet_cpu}, we show the CPU timings of the algorithms for the 13 scenarios in the CTU-13 dataset under both hypothesis, where the algorithms are run on the same processor described in the simulations. Even for scenario 8, the computational times of MRE-HP are feasible despite running on a rather out-of-date processor with a low clock speed. Again we mark NA for the scenarios that are computationally infeasible for the TA-Map method (the memory requirements are above 32GB even for scenario 6). The MRE-HP method despite being slower than the TA-Map on smaller networks (see Table \ref{Fig7_cpu}), scales much more efficiently to larger networks.

\setlength\extrarowheight{2pt}
\setlength{\tabcolsep}{5pt}
\begin{table}[h]
	\centering
	\caption{CTU Network CPU Times (in seconds)}
	\label{botnet_cpu}
	\begin{tabular}{|c|ccc|ccc|} \hline
		\hspace{-6pt} Scenario \hspace{-6pt} & \multicolumn{3}{c|}{When $\mathcal{H}_0$ is True }  & \multicolumn{3}{c|}{When $\mathcal{H}_A$ is True } \\ 
		& MRE-HP & \hspace{-4pt} MLE \hspace{-4pt} & TA-Map & MRE-HP & \hspace{-4pt} MLE \hspace{-4pt} & TA-Map\\ \hline
		1 & 513.72 & 25.381 & NA & 1303.2 & 62.148 & NA \\ 
		2 & 19.258 & 5.2586 & 426.69 & 206.71  & 1.0702 & 683.26 \\ 
		3 & 790.60 & 70.706 & NA & 468.39 & 55.564 & NA \\ 
		4 & 2038.0 & 10.834 & NA & 3607.5 & 30.863 & NA \\
		5 & 539.85 & 2.0568 & NA & 263.56 & 3.2602 & NA \\ 
		6 & 69.452 & 1.6095 & NA & 58.427 & 12.556 & NA \\ 
		7 & 10.164 & 0.8487 & 360.09 & 17.043 & 0.6761 & 366.83 \\ 
		8 & 62602 & 8071.2 & NA & 55591 & 2087.8 & NA \\ 
		9 &  5439.3 & 97.082 & NA & 903.46 & 51.762  & NA \\ 
		10 & 648.88 & 7.8440 & NA & 174.66 & 2.9925 & NA \\ 
		11 & 4.2550 & 0.7101 & 55.126 & 17.645 & 0.4735 & 56.367 \\ 
		12 & 1864.8 & 3.6154 & NA & 514.97 & 5.5562 & NA \\ 
		13 & 355.20 & 17.620 & NA & 792.81 & 76.237 & NA \\ \hline
	\end{tabular}
\end{table}

\subsection{Taxi Dataset}

The proposed model was applied to a dataset consisting of yellow and green taxicabs rides from the New York City Taxi and Limousine Commission (NYC TLC) \cite{NYC_TLC} and \cite{taxi_github}. For every NYC taxicab ride, the dataset contains the pickup and drop-off locations as geographic coordinates (latitude and longitude). Green taxicabs are not allowed to pickup passengers below West 110th Street and East 96th Street in Manhattan, but occasionally they risk the chance of getting punished and ignore the regulations. In an article on June 10th 2014, the New York Post explains how the city began hiring more TLC inspectors to catch illegal pickups and enforce the location rules \cite{NY_post}. Thus we are interested in identifying if there are green taxicabs operating in lower Manhattan when we only know the yellow taxicab network. We treat the 18 Neighborhood Tabulation Areas (NTA) in lower Manhattan as nodes and associate any pickups or drop-offs within an NTA's boundaries as traffic entering or leaving the node. We form edges from only frequently occurring routes of traffic, which we define as having activity at least an average of every 20 minutes for yellow taxicabs and twice a month for green taxicabs. For samples, we use the yellow and green taxicab rides from between January and May of 2014 and aggregate them into daily totals. 

Like in the previous example, we indirectly observe samples of the total ingress and egress of each node, and the total flows of 10 interior nodes that each observe the flows of $0.7 P$ nodes. This creates a total traffic network with $P = 18$ nodes and $187$ non-zero edges ($39\%$ sparsity) where the baseline network (yellow taxicab rides) has $163$ of the edges. There is anomalous activity (green taxicab rides) on $56$ of the edges, where $32$ of these edges are also in the baseline network and $24$ are not. We observe the network for a total of $T = 150$ days. Fig. \ref{fig:taxi_network} shows the baseline network formed from yellow taxicab rides and the unknown anomalous activity due to illegal pickups from green taxicabs. 

\begin{figure}[h] 
	\centering
	\includegraphics[width=\linewidth]{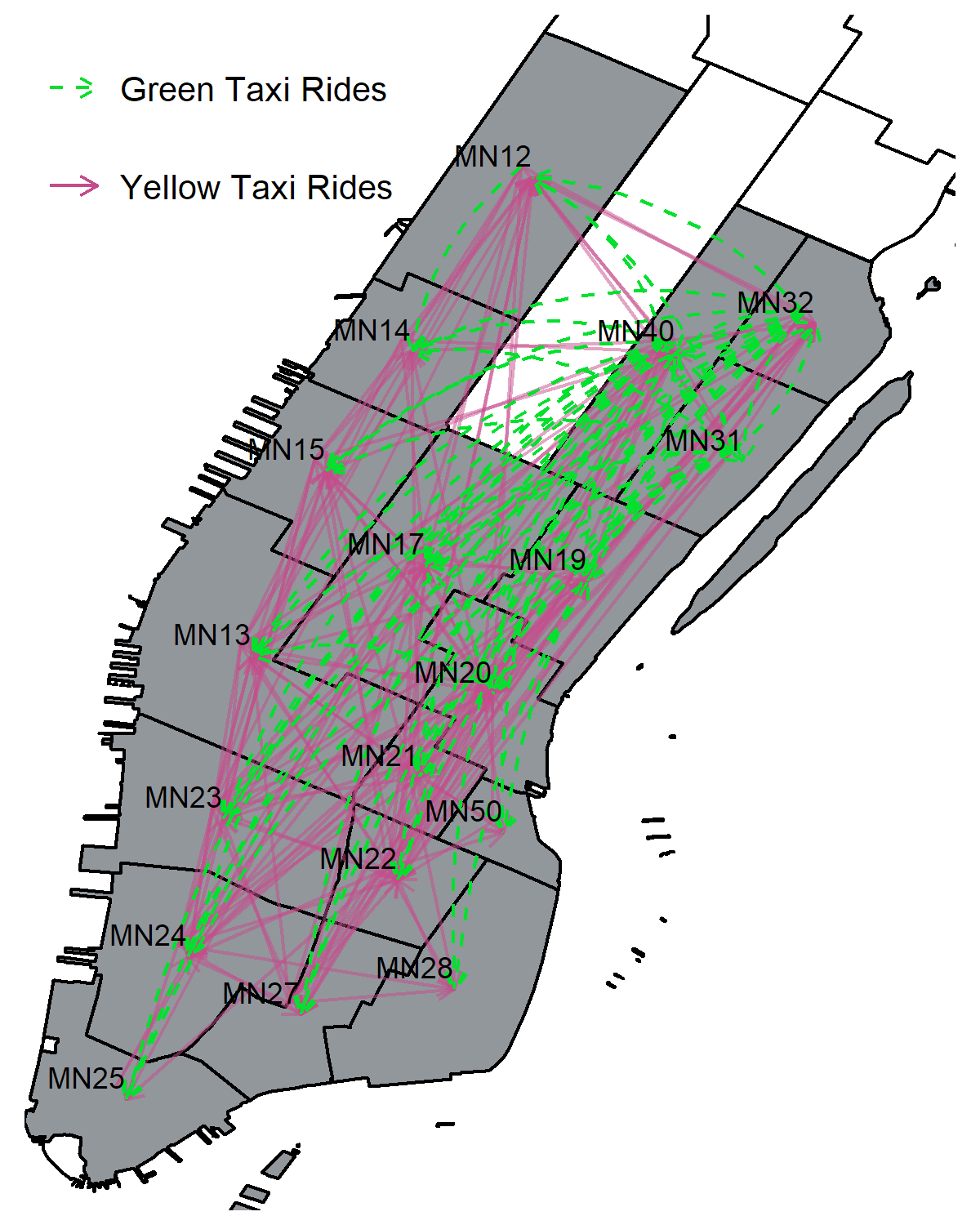}
	\caption{A network of taxicab rides in lower Manhattan where the nodes are the 18 NTAs. The traffic from yellow taxicab rides (solid purple lines) form the baseline network and the traffic from green taxicab rides (dashed green lines) are anomalous activity in the network.}
	\label{fig:taxi_network}
\end{figure}

Table \ref{taxi_test} shows, for different percentages of edges observed, whether the correct decision (reject or not) is made when the null hypothesis is true (no green taxi traffic) and when it is not true (green taxi traffic). The Hierarchical Poisson model initialized at the MRE solution always makes the correct decision while the Poisson MLE model, except for when the network can be directly observed, always rejects the null hypothesis. These two models are tested at an alpha level of $0.05$. The Traffic and Anomaly Map method, which has a $0.05$ Type-I error threshold chosen from the ROC curves of the simulations, also has poor performance.

\begin{table}[h]
	\centering
	\caption{Taxi Network Test}
	\label{taxi_test}
	\begin{tabular}{|c|ccc|ccc|} \hline
		\% & \multicolumn{3}{c|}{When $\mathcal{H}_0$ is True }  & \multicolumn{3}{c|}{When $\mathcal{H}_A$ is True } \\ 
		\hspace{-4pt} Edges \hspace{-4pt} & MRE-HP & \hspace{-4pt} MLE \hspace{-4pt} & TA-Map & MRE-HP & \hspace{-4pt} MLE \hspace{-4pt} & TA-Map \\ \hline
		0 & \checkmark & $\times$ & \checkmark & \checkmark & \checkmark & $\times$ \\ 
		10\% & \checkmark & $\times$ & \checkmark & \checkmark & \checkmark & $\times$ \\ 
		20\% & \checkmark & $\times$ & \checkmark & \checkmark & \checkmark & $\times$ \\ 
		30\% & \checkmark & $\times$ & \checkmark & \checkmark & \checkmark & $\times$ \\
		40\% & \checkmark & $\times$ & \checkmark & \checkmark & \checkmark & $\times$  \\ 
		50\% & \checkmark & $\times$ & \checkmark & \checkmark & \checkmark & $\times$ \\ 
		60\%& \checkmark & $\times$ & $\times$ & \checkmark & \checkmark & $\times$ \\ 
		70\%& \checkmark & $\times$ & $\times$ & \checkmark & \checkmark & \checkmark \\ 
		80\%&\checkmark & $\times$ &  $\times$ & \checkmark & \checkmark & \checkmark\\ 
		90\%& \checkmark & $\times$ &  $\times$ & \checkmark & \checkmark & \checkmark \\ 
		100\%& \checkmark & \checkmark & $\times$ & \checkmark & \checkmark & \checkmark \\ \hline
	\end{tabular}
\end{table}

From the results of Table \ref{taxi_test}, we know the Hierarchical Poisson model initialized at the MRE solution is always able to detect changes in the network at a global scale, but we are also interested in the recovery of the individual green taxicab routes. When 70\% of the network is observed, the model is able to detect 52 of the 56 edges that contain anomalous activity with only a 2\% false positive rate. The 4 missed edges and 5 false alarms are shown in Fig. \ref{fig:edge_errors}. 

\begin{figure}[h] 
	\centering
	\includegraphics[width=\linewidth]{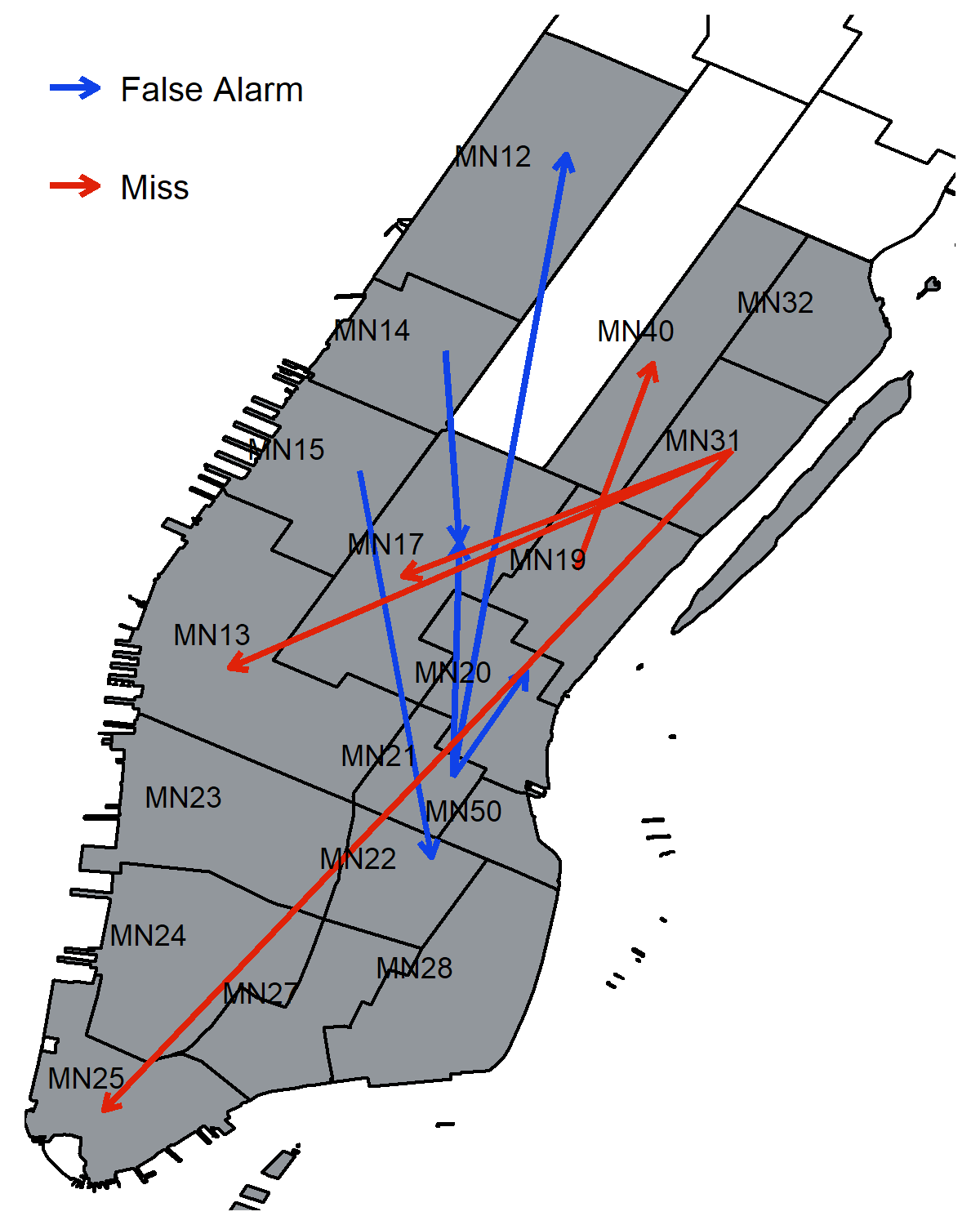}
	\caption{A miss (red line) is an edge that the MRE-HP model fails to identify as containing anomalous activity and a false alarm (blue line) is an edges that is incorrectly identified as containing anomalous activity. The majority of the misses depart from MN31 (Lenox Hill and Roosevelt Island), which may contain legal activity because green taxis are allowed to pick up passengers from Roosevelt Island.}
	\label{fig:edge_errors}
\end{figure}

Out of the 4 misses, 3 of them are from green taxicab pickups from MN31, which contains the Lenox Hill and Roosevelt Island areas. Green taxicabs are allowed to pick up passengers from Roosevelt Island, but not from Lenox Hill, so some of the traffic on these 3 routes could be legal and not anomalous activity. The other miss, from MN19 to MN40, only had 11 rides in 150 days, making it harder to distinguish from just perturbation noise in the samples.

\section{Conclusion}

We have developed a framework and a probabilistic model for detecting anomalous activity in the traffic rates of sparse networks. Our framework is realistic and robust in that, at minimum, it only requires observing the total egress and ingress of the nodes. Because it imposes no fixed assumptions of edge structure, our framework allows the estimator to handle noisy observations and anomalous activity. Our simulation results show the advantages of our model over competing methods in detecting anomalous activity. Through application of our models to the CTU-13 botnet datasets, we show that the model is scalable and robust to various scenarios, and with the NYC taxi dataset, we show an application of our model and framework to an already identified real-world problem.

\appendix

\begin{proof}[Proof of Proposition \ref{prop:indep_time_LB}]
By Jensen's inequality, $ \log \left( \P(\mathcal{D}| \bm{\Lambda}) \right) $
\begin{flalign*} 	
& = \log \left( \prod_{t_{1}=1}^T \P(\bm{R}^{t_{1}}| \bm{\Lambda}) \prod_{t_{2}=1}^T \P(\bm{C}^{t_{2}}| \bm{\Lambda})  \prod_{t_{3}=1}^T \P(\bm{F}^{t_{3}}| \bm{\Lambda}) \right) \\
& \geq \sum_{t_{1}=1}^T \E_{\q^{t_{1}}} \left(\log \P(\bm{R}^{t_{1}} \hspace{-4pt} ,\bm{N}^{t_{1}} | \bm{\Lambda})\right) - \E_{\q^{t_{1}}} \left( \log \q(\bm{N}^{t_{1}}) \right) \\
& + \sum_{t_{2}=1}^T \E_{\q^{t_{2}}} \left( \log \P(\bm{C}^{t_{2}} \hspace{-4pt} , \bm{N}^{t_{2}} | \bm{\Lambda}) \right) -\E_{\q^{t_{2}}} \left( \log \q(\bm{N}^{t_{2}}) \right) \\
& + \sum_{t_{3}=1}^T \E_{\q^{t_{3}}} \left( \log \P(\bm{F}^{t_{3}} \hspace{-4pt} , \bm{N}^{t_{3}}| \bm{\Lambda}) \right) -\E_{\q^{t_{3}}} \left( \log \q(\bm{N}^{t_{3}}) \right) \\
& = \sum_{t_{1}=1}^T \E_{\q^{t_{1}}} \left(\log \P(\bm{R}^{t_{1}} | \bm{N}^{t_{1}} \hspace{-4pt} , \bm{\Lambda})\right) +\sum_{t_{2}=1}^T \E_{\q^{t_{2}}}  \left( \log \P(\bm{C}^{t_{2}}| \bm{N}^{t_{2}} \hspace{-4pt} , \bm{\Lambda}) \right) \\
& + \sum_{t_{3}=1}^T \E_{\q^{t_{3}}} \left( \log \P(\bm{F}^{t_{3}}| \bm{N}^{t_{3}} \hspace{-4pt} , \bm{\Lambda}) \right) \\
& + \sum_{\tau=1}^3 \sum_{t_{\tau}=1}^T \E_{\q^{t_{\tau}}} \left( \log \P( \bm{N}^{t_{\tau}}| \bm{\Lambda}) \right) +  \text{H} \left( \q(\bm{N}^{t_{\tau}}) \right) 
\end{flalign*}
and $ \P(\bm{R}^{t_{1}} \hspace{-2pt} | \bm{N}^{t_{1}} \hspace{-6pt} , \bm{\Lambda}) = \P(\bm{C}^{t_{2}} \hspace{-2pt} | \bm{N}^{t_{2}} \hspace{-6pt} , \bm{\Lambda}) =\P(\bm{F}^{t_{3}} \hspace{-2pt} | \bm{N}^{t_{3}} \hspace{-6pt} , \bm{\Lambda}) = 1 $.
The inequality is tight (by KL divergence) when $ \q(\bm{N}^{t_{1}}) = \P(\bm{N}^{t_{1}} | \bm{R}^{t_{1}} \hspace{-4pt} , \bm{\Lambda}) $, $\q(\bm{N}^{t_{2}}) = \P(\bm{N}^{t_{2}} | \bm{C}^{t_{2}} \hspace{-4pt} , \bm{\Lambda})$, and $ \q(\bm{N}^{t_{3}}) = \P(\bm{N}^{t_{3}} | \bm{F}^{t_{3}} \hspace{-4pt} , \bm{\Lambda}) $ are multinomial distributions.
\end{proof}

\begin{proof}[Proof of Theorem \ref{thm:hbayes}] 
$ $ \\
Define $\mathcal{N} = \left\{\bm{N}^{t_{\tau}}  \hspace{-4pt} : \forall t_{\tau} = 1, \dots, T \text{ and } \tau = 1, \dots, 3 \right\} $ as the set of all network traffic at different time points $t_{\tau}$ for the entire sample window $1, \dots, T$. So $\cap \, \mathcal{N}$ is the intersection of the set and $\P(\cap \, \mathcal{N})$ is its joint probability.
By Jensen's inequality, $ \log \P( \bm{\Lambda} | \mathcal{D} )  $	
\begin{flalign} \label{lower_bound}
& = \log \int \P(\bm{\Lambda}, \bm{\epsilon} | \mathcal{D} ) \, d\bm{\epsilon} = \log \int \frac{\P(\mathcal{D} | \bm{\Lambda}, \bm{\epsilon}) \P(\bm{\Lambda} | \bm{\epsilon}) \P(\bm{\epsilon}) }{\P(\mathcal{D}) } \, d\bm{\epsilon} \notag \\
& = \log\int \left( \int \cdots \int \P \left(\mathcal{D}, \cap \, \mathcal{N} | \bm{\Lambda}, \bm{\epsilon} \right) \frac{ \P(\bm{\Lambda} | \bm{\epsilon}) \P(\bm{\epsilon} )}{\P(\mathcal{D})} d\mathcal{N} \right) d\bm{\epsilon} \notag \\
& = \log\int  \E_{\q} \left( \frac{ \P(\mathcal{D}, \cap \, \mathcal{N} | \bm{\Lambda}, \bm{\epsilon} ) }{ \prod_{\tau=1}^3 \prod_{t_{\tau}=1}^T \q(\bm{N}^{t_{\tau}}) } \frac{ \P(\bm{\Lambda} | \bm{\epsilon}) \P(\bm{\epsilon} )}{\P(\mathcal{D})} \right) d\bm{\epsilon} \notag \\
& \ge \log \int \exp \bigg\{  \E_{\q} \bigg( \log \P(\mathcal{D}, \cap \, \mathcal{N} | \bm{\Lambda}, \bm{\epsilon} ) + \log \P(\bm{\Lambda} | \bm{\epsilon}) \notag \\
& \hspace{72pt} + \log \frac{ \P(\bm{\epsilon} )}{\P(\mathcal{D})} - \sum_{\tau=1}^3 \sum_{t_{\tau}=1}^T \log \q(\bm{N}^{t_{\tau}}) \bigg) \bigg\} d\bm{\epsilon} \notag \\
& = \log \hspace{-3pt} \int \hspace{-3pt} \exp \left\{  \E_{\q} \left( \log \frac{\P(\mathcal{D}, \cap \, \mathcal{N}, \bm{\Lambda}| \bm{\epsilon}) }{\P(\mathcal{D}, \cap \, \mathcal{N} | \bm{\epsilon}) } \frac{\P(\mathcal{D}, \cap \, \mathcal{N} | \bm{\epsilon}) \P(\bm{\epsilon} )}{\P(\mathcal{D})} \right) \right\} d\bm{\epsilon} \notag \\
& + \log \left( \exp \left\{ -  \E_{\q} \left( \sum_{\tau=1}^3 \sum_{t_{\tau}=1}^T \log \q(\bm{N}^{t_{\tau}}) \right) \right\} \right) \notag \\
& = \log \hspace{-3pt} \int \hspace{-3pt} \exp \left\{  \E_{\q} \left( \log \P( \bm{\Lambda} | \cap \, \mathcal{N} , \bm{\epsilon}) \right) +  \E_{\q} \left( \log \P( \cap \, \mathcal{N} , \bm{\epsilon} | \mathcal{D}) \right) \right\} \hspace{-2pt} d\bm{\epsilon} \notag \\
& + \sum_{\tau=1}^3 \sum_{t_{\tau}=1}^T \text{H} \left( \q(\bm{N}^{t_{\tau}}) \right)  \\
& = \log \hspace{-3pt} \int \hspace{-3pt} \exp \left\{  \E_{\q} \left( \log \P( \bm{\Lambda} | \cap \, \mathcal{N} , \bm{\epsilon}) \right) + \E_{\q} \left( \log \P(\bm{\epsilon} | \cap \, \mathcal{N} ) \right) \right\} \hspace{-2pt} d\bm{\epsilon} \notag
\end{flalign}
where this bound is tight (by KL divergence) when $\q = \P( \cap \, \mathcal{N} | \mathcal{D}, \bm{\Lambda}, \bm{\epsilon} ) $ \\ $= \prod_{\tau = 1}^3 \prod_{t_{\tau}=1}^T \P(\bm{N}^{t_{\tau}} | \bm{R}^{t_{\tau}}, \bm{\Lambda}) \P(\bm{N}^{t_{\tau}} | \bm{C}^{t_{\tau}}, \bm{\Lambda}) \P(\bm{N}^{t_{\tau}} | \bm{F}^{t_{\tau}}, \bm{\Lambda}) $ are multinomial distributions. And, maximizing \\ $ \E_{\q} \left( \log \P(\cap \, \mathcal{N}, \bm{\epsilon} | \mathcal{D}) \right) $
 \begin{flalign*} 
 & =  \E_{\q} \left( \log \P( \mathcal{D}|\cap \, \mathcal{N} , \bm{\epsilon} ) + \log \P(\cap \, \mathcal{N} , \bm{\epsilon} ) - \log \P(\mathcal{D}) \right) \\
 & =  \E_{\q} \left( \log(1) + \log \P(\cap \, \mathcal{N} | \bm{\epsilon} ) \right) + \log \P(\bm{\epsilon} ) - \log \P(\mathcal{D}) \\
 & = \log \P(\bm{\epsilon} ) - \log \P(\mathcal{D}) + \sum_{\tau=1}^3 \sum_{t_{\tau}=1}^T \E_{\q^{t_{\tau}}} \log \P(\bm{N}^{t_{\tau}} | \bm{\epsilon})
 \end{flalign*}
 is equivalent to maximizing a lower bound of $\log \P( \bm{\epsilon} | \mathcal{D} ) $
 \begin{flalign*} 
 & = \log \frac{\prod_{t_{1}=1}^T \prod_{t_{2}=1}^T \prod_{t_{3}=1}^T \P(\bm{R}^{t_{1}}, \bm{C}^{t_{2}}, \bm{F}^{t_{3}}| \bm{\epsilon}) \P(\bm{\epsilon}) }{\P(\mathcal{D})} \\
 & = \log \P(\bm{\epsilon}) - \log \P(\mathcal{D}) + \log \prod_{t=1}^T \E_{\q^{t_{1}}} \left(\frac{ \P(\bm{R}^{t_{1}} \hspace{-4pt},\bm{N}^{t_{1}}| \bm{\epsilon}) }{ \q^{t_{1}}(\bm{N}^{t_{1}}) } \right) \\
 & + \log \hspace{-2pt} \prod_{t_{2}=1}^T \hspace{-2pt}  \E_{\q^{t_{2}}} \hspace{-2pt} \left( \frac{\P(\bm{C}^{t_{2}} \hspace{-4pt}, \bm{N}^{t_{2}}| \bm{\epsilon}) }{ \q^{t_{2}}(\bm{N}^{t_{2}}) } \right) \hspace{-2pt} + \log \hspace{-2pt} \prod_{t_{3}=1}^T \hspace{-2pt} \E_{\q^{t_{3}}} \hspace{-2pt} \left(\frac{ \P(\bm{F}^{t_{3}} \hspace{-4pt}, \bm{N}^{t_{3}}| \bm{\epsilon}) }{ \q^{t_{3}}(\bm{N}^{t_{3}}) } \right) \\
 & \geq \log \P(\bm{\epsilon}) - \log \P(\mathcal{D}) + \sum_{{t_{1}}=1}^T \E_{\q^{t_{1}}} ( \log \P(\bm{R}^{t_{1}} | \bm{N}^{t_{1}} ) \hspace{-4pt}, \bm{\epsilon})\\
 & + \sum_{t_{2}=1}^T \E_{\q^{t_{2}}} \left( \log \P(\bm{C}^{t_{2}} | \bm{N}^{t_{2}} \hspace{-4pt}, \bm{\epsilon}) \right) + \sum_{t_{3}=1}^T \E_{\q^{t_{3}}} \left( \log \P( \bm{F}^{t_{3}} | \bm{N}^{t_{3}} \hspace{-4pt} ,\bm{\epsilon}) \right) \\
 & + \sum_{\tau=1}^3 \sum_{t^{(\tau)}=1}^T \E_{\q^{t_{\tau}}} \left( \log \P(\bm{N}^{t_{\tau}} | \bm{\epsilon}) \right) - \E_{\q^{t_{\tau}}} \left( \log \q(\bm{N}^{t_{\tau}}) \right) \\
& \propto \log \P(\bm{\epsilon} ) - \log \P(\mathcal{D}) + \sum_{\tau=1}^3 \sum_{t_{\tau}=1}^T \E_{\q^{t_{\tau}}} \left( \log \P(\bm{N}^{t_{\tau}} | \bm{\epsilon}) \right) 
 \end{flalign*}
 for any distributions of $ \q(\bm{N}^{t_{1}}) , \q(\bm{N}^{t_{2}}) , \q(\bm{N}^{t_{3}}) $. \\
 
 \noindent Since $ \bm{N}_{ij}^{t_{\tau}} | \epsilon_{ij} \sim NegBin(\epsilon_{ij} \Lambda_{0 \, ij} + 1, \frac{1}{1+\epsilon_{ij}}) $ is the negative binomial distribution and $\epsilon_{ij} \sim Unif(0, \infty)$, the M-step is $\hat{\epsilon_{ij}}$
 \begin{flalign*} 
 & = \underset{\epsilon_{ij}}{\arg\max} \, \log \P(\epsilon_{ij}) + \sum_{\tau=1}^3 \sum_{t_{\tau}=1}^T \E_{\q^{t_{\tau}}} \left( \log \P(\bm{N}_{ij}^{t_{\tau}} | \epsilon_{ij}) \right) \\
 & \propto \underset{\epsilon_{ij}}{\arg\max} \, \sum_{\tau=1}^3 \sum_{t_{\tau}=1}^T \E_{\q^{t_{\tau}}} \left( \log \Gamma(N_{ij}^{t_{\tau}}+ \epsilon_{ij} \Lambda_{0 \, ij} + 1) \right) \\
 & + \log(\epsilon_{ij}) 3T (\epsilon_{ij} \Lambda_{0 \, ij} + 1) - \log(1+\epsilon_{ij}) 3T (\epsilon_{ij} \Lambda_{0 \, ij} + 1) \\
 &- 3T \log \Gamma(\epsilon_{ij} \Lambda_{0 \, ij} + 1) -\log(1+\epsilon_{ij}) \sum_{\tau=1}^3 \sum_{t_{\tau}=1}^T \E_{\q^{t_{\tau}}} (N_{ij}^{t_{\tau}}) \\
 & \ge \underset{\epsilon_{ij}}{\arg\max} \, 3T \left( \hspace{-2pt} (\epsilon_{ij} \Lambda_{0 \, ij} + 1) \log \frac{\epsilon_{ij}}{ 1+\epsilon_{ij}} - \log \Gamma(\epsilon_{ij} \Lambda_{0 \, ij} + 1) \hspace{-2pt} \right) \\
 & + \sum_{\tau=1}^3 \sum_{t_{\tau}=1}^T \log \Gamma(\E_{\q^{t_{\tau}}} (N_{ij}^{t_{\tau}}) + \epsilon_{ij} \Lambda_{0 \, ij} + 1) \\
 & -\log(1+\epsilon_{ij}) \sum_{\tau=1}^3 \sum_{t^{\tau}=1}^T \E_{\q^{t_{\tau}}} (N_{ij}^{t_{\tau}}) 
\end{flalign*}
and given estimates of the hyperparameters $\hat{\epsilon}_{ij}$, estimators for the rates $\hat{\Lambda}_{ij}$
 \begin{flalign*} 
 & = \underset{\Lambda_{ij}}{\arg\max} \,  \E_{\q} \left( \log \P( \cap \, \mathcal{N} | \bm{\Lambda}, \hat{\bm{\epsilon}}) + \log \P( \bm{\Lambda}| \hat{\bm{\epsilon}}) - \log \P( \cap \, \mathcal{N} ) \right) \\
 & \propto \underset{\Lambda_{ij}}{\arg\max} \, \log \P(\bm{\Lambda} | \hat{\bm{\epsilon}}) + \sum_{\tau=1}^3 \sum_{t_{\tau}=1}^T \E_{\q^{t_{\tau}}} \left( \log \P(\bm{N}^{t_{\tau}} | \bm{\Lambda}) \right) \\
 & \propto \underset{\Lambda_{ij}}{\arg\max} \, (\hat{\epsilon}_{ij} \Lambda_{0\, ij} ) \log(\Lambda_{ij}) - \hat{\epsilon}_{ij} \Lambda_{ij} -3T \Lambda_{ij} \\
 & + \sum_{\tau=1}^3 \sum_{t_{\tau}=1}^T \E_{\q^{t_{\tau}}} (N_{ij}^{t_{\tau}}) 
 \end{flalign*}
 
Thus when $\E_{\q^{t_{1}}} (N_{ij}^t) = \E(N_{ij}^{t_{1}} | \bm{R}^{t_{1}} , \hat{\bm{\Lambda}}^{k}) $ where $\hat{\bm{\Lambda}}^{k}$ are the previous iterations' estimators for the rate matrix, the lower bound will push up against the observed log posterior $ \log \P(\bm{\Lambda} | \mathcal{D})$. This makes the E-step just the means of the independent Multinomial distributions $ \prod_{i=1}^P Multi(R_{i}^{t_{1}}, \frac{\hat{\Lambda}_{i1}^{k}}{\sum_{j=1}^P \hat{\Lambda}_{ij}^{k} }, \dots, \frac{\hat{\Lambda}_{iP}^{k}}{\sum_{j=1}^P \hat{\Lambda}_{ij}^{k} })$ like in the previous models. The same holds when given the column sums $ \bm{C}^{t_{2}} $ or flows $ \bm{F}^{t_{3}} $.
\end{proof}

\begin{proof}[Proof of Proposition \ref{thm:mre}]
	
	The positive estimator $\hat{\bm{\Lambda}} $ that maximizes the MRE distribution is the solution to
	\begin{flalign*} 
	& = \underset{\bm{\Lambda} \in \mathbb{R}^+ }{\arg\max} \, \log \left( \text{P}(\bm{\Lambda} | \bm{R}, \bm{C}, \bm{F}) \right) \\
	& = \underset{\bm{\Lambda} \in \mathbb{R}^+ }{\arg\max} \, \log ( \prod_{ij} \exp \left\{ -| \Lambda_{ij} - \Lambda_{0ij} | \right\} ) -\log(Z\left(\bm{\rho}, \bm{\gamma}, \bm{\phi}) \right) \\
	& + \log (\exp\{ \hat{\bm{\rho}}' (\bm{\Lambda} \bm{1} - \bar{\bm{R}}) + \hat{\bm{\gamma}}' (\bm{1}' \bm{\Lambda} - \bar{\bm{C}}) + \hat{\bm{\phi}}' (\bm{A} \bm{\Lambda} \bm{B} - \bar{\bm{F}})\}) \\
	& = \underset{\bm{\Lambda} \in \mathbb{R}^+ }{\arg\max} \, - \sum_{ij} | \Lambda_{ij} - \Lambda_{0ij} | + \hat{\bm{\rho}}' (\bm{\Lambda} \bm{1} - \bar{\bm{R}}) + \hat{\bm{\gamma}}' (\bm{\Lambda}'\bm{1} - \bar{\bm{C}}) \\
	&  + \hat{\bm{\phi}}' (\bm{A} \bm{\Lambda} \bm{B} - \bar{\bm{F}}) \\
	& = \underset{\bm{\Lambda} \in \mathbb{R}^+ }{\arg\min} \, ||\bm{\Lambda} - \bm{\Lambda}_0||_1 - \hat{\bm{\rho}}' (\bm{\Lambda} \bm{1} - \bar{\bm{R}}) - \hat{\bm{\gamma}}' (\bm{\Lambda}' \bm{1} - \bar{\bm{C}}) \\
	& - \hat{\bm{\phi}}' (\bm{A} \bm{\Lambda} \bm{B} - \bar{\bm{F}}) 
	\end{flalign*}
	where $|| \cdot ||_1$ is the element wise $\ell_1$ norm and the optimal Lagrange multipliers $\hat{\bm{\rho}}, \hat{\bm{\gamma}}, \hat{\bm{\phi}} $ are the solution to
	\begin{flalign} \label{dual}
	& = \underset{\bm{\rho}, \bm{\gamma}, \bm{\phi}}{\arg\max} \, -\log \left(Z(\bm{\rho}, \bm{\gamma}, \bm{\phi}) \right) \\
	& = \underset{\bm{\rho}, \bm{\gamma}, \bm{\phi}}{\arg\max} \, \sum_{i=1}^P \rho_i \bar{R}_i +\sum_{j=1}^P \gamma_j \bar{C}_j +\sum_{h} \phi_h \bar{F}_h - \log2 \notag \\
	& - \sum_{ij} \Lambda_{0ij}(\rho_i+\gamma_j+ \sum_{h}\phi_h A_{hi} B_j) \notag \\
	& +\log(1+LM_{ij})+\log(1-LM_{ij} ) \notag \\
	& = \underset{\bm{\rho}, \bm{\gamma}, \bm{\phi}}{\arg\max} \, \sum_{i=1}^P \rho_i (\bar{R}_i - \sum_{j=1}^P \Lambda_{0ij}) + \sum_{j=1}^P \gamma_j (\bar{C}_j - \sum_{i=1}^P \Lambda_{0ij}) \notag \\
	& + \sum_{h} \phi_h ( \bar{F}_h - \sum_{ij} A_{hi} \Lambda_{0ij} B_j) +\sum_{ij} \log(1 -LM_{ij}^2) \notag 
	\end{flalign} 
	where $LM_{ij} = \rho_i+\gamma_j+\sum_{h}\phi_h A_{hi} B_j$.
	
	The Lagrangian of the loss function in \eqref{l1-opt} is $ ||\bm{\Lambda} - \bm{\Lambda}_0||_1 \\ + \bm{\rho}' (\bm{\Lambda} \bm{1} - \bar{\bm{R}}) + \bm{\gamma}' (\bm{1}' \bm{\Lambda} - \bar{\bm{C}}) + \bm{\phi}' (\bm{A} \bm{\Lambda} \bm{B} - \bar{\bm{F}}) $
	with optimal Lagrange multipliers that are the solution to dual problem
	\begin{flalign*} 
	& = \underset{\bm{\rho}, \bm{\gamma}, \bm{\phi}}{\arg\max} \, - \sum_{ij} f^*(-\rho_i-\gamma_j-\sum_{h}\phi_h A_{hi} B_j) - \sum_{i=1}^P \rho_i \bar{R}_i  \\
	& - \sum_{j=1}^P \gamma_j \bar{C}_j - \sum_{h} \phi_h \bar{F}_h \\
	& = \underset{\bm{\rho}, \bm{\gamma}, \bm{\phi}}{\arg\max} \, \sum_{ij} \Lambda_{0ij}(LM_{ij} ) - \sum_{i=1}^P \rho_i \bar{R}_i - \sum_{j=1}^P \gamma_j \bar{C}_j  - \sum_{h} \phi_h \bar{F}_h \\
	& \hspace{66pt} \text{ subject to } | LM_{ij} | < 1 \,\, \forall i,j
	\end{flalign*} 
	because $ f^*(-\rho_i-\gamma_j-\sum_{h}\phi_h A_{hi} B_j) $ are the convex conjugates defined as
	\begin{flalign*} 
	& = \underset{\Lambda_{ij}}{\max} -\Lambda_{ij}(\rho_i+\gamma_j+\sum_{h}\phi_h A_{hi} B_j) - | \Lambda_{ij} - \Lambda_{0ij} | \\
	& = \underset{\Lambda_{ij}}{\max} \begin{cases}
	&\Lambda_{0ij} -\Lambda_{ij}(1 + LM_{ij} ) \qquad \text{ if } \quad \quad \Lambda_{ij} \geq \Lambda_{0ij} \\
	& \Lambda_{ij} (1 - LM_{ij} ) - \Lambda_{0ij} \qquad \text{ if } \quad \quad \Lambda_{ij} < \Lambda_{0ij}
	\end{cases} \\
	& = \begin{cases}
	& \infty \hspace{16pt} \text{ if } \quad \qquad | \rho_i+\gamma_j+\sum_{h}\phi_h A_{hi} B_j | > 1 \\
	& -\Lambda_{0ij}(\rho_i+\gamma_j+\sum_{h}\phi_h A_{hi} B_j) \qquad \text{ otherwise}.
	\end{cases}
	\end{flalign*}
	The dual can be relaxed with log barrier terms to an unconstrained problem that is equivalent to \eqref{dual} making minimizing the Lagrangian of \eqref{l1-opt} for $\bm{\Lambda}$ equivalent to maximizing the MRE distribution.
	
\end{proof}

\begin{proof}[Proof of Proposition \ref{prop:pdr}] 
Using Remark 1.7 of \cite{watanabe2009algebraic}, then for regular models, the MAP estimator will have the same asymptotic properties as the MLE. Thus, the standard proof for the asymptotic distribution for the log likelihood ratio \cite{wilks} applies to the log posterior density ratio.
\end{proof}

\begin{proof}[Proof of Proposition \ref{prop:miss}] 
Let $\mathcal{M}(\bm{\Lambda}^*)$ be the true model, then the test statistic $\psi$
\begin{flalign*}
& =-2 \sum_{t=1}^T \log(\mathcal{M}_t(\bm{\Lambda}_0) ) - \log(\mathcal{M}_t(\hat{\bm{\Lambda}})) \\
& = -2 \left( \sum_{t=1}^T \log( \mathcal{M}_t(\bm{\Lambda}_0) ) - \underset{\bm{\Lambda} \in \mathbb{R}^+}{\max} \sum_{t=1}^T \log(\mathcal{M}_t(\bm{\Lambda})) \right)\\
& = 2 \sum_{t=1}^T \log(\mathcal{M}_t(\bm{\Lambda}^*) ) - \log( \mathcal{M}_t(\bm{\Lambda}_0) ) \\
& -2 \, \underset{\bm{\Lambda} \in \mathbb{R}^+}{\min} \sum_{t=1}^T \log(\mathcal{M}_t(\bm{\Lambda}^*) ) - \log(\mathcal{M}_t(\bm{\Lambda})) \\
& \text{and as } T \rightarrow \infty, \, \psi / T \\
& \rightarrow 2 \, \text{KL}\left( \mathcal{M}(\bm{\Lambda}^* || \mathcal{M}(\bm{\Lambda}_0) \right) -2 \underset{\bm{\Lambda} \in \mathbb{R}^+}{\min} \, \text{KL} \left( \mathcal{M}(\bm{\Lambda}^*) || \mathcal{M}(\bm{\Lambda}) \right) \\
& = 2 \, \text{KL}\left( \mathcal{M}(\bm{\Lambda}^* || \mathcal{M}(\bm{\Lambda}_0) ) \right) = \Psi
\end{flalign*}

The misspecified test statistic $\hat{\psi}$
\begin{flalign} \label{decompose}
 & =-2 \sum_{t=1}^T \log( \hat{\mathcal{M}}_t^k(\bm{\Lambda}_0) ) - \log(\hat{\mathcal{M}}_t^k(\hat{\bm{\Lambda}})) \notag \\
& = -2 \sum_{t=1}^T \log( \hat{\mathcal{M}}_t^k(\bm{\Lambda}_0) ) - \underset{\bm{\Lambda} \in \mathbb{R}^+}{\max} \sum_{t=1}^T \log(\hat{\mathcal{M}}_t^k(\bm{\Lambda})) \notag \\
& = 2 \sum_{t=1}^T \log(\mathcal{M}_t(\bm{\Lambda}^*) ) - \log( \mathcal{M}_t(\bm{\Lambda}_0) ) \\
& \hspace{12pt} +2 \sum_{t=1}^T \log(\mathcal{M}_t(\bm{\Lambda}_0) ) - \log( \hat{\mathcal{M}}_t^k(\bm{\Lambda}_0) ) \notag \\
& \hspace{12pt} -2 \sum_{t=1}^T \log(\mathcal{M}_t(\bm{\Lambda}^*) ) - \log( \hat{\mathcal{M}}_t^k(\hat{\bm{\Lambda}}^*) ) \notag \\
& \hspace{12pt} -2 \underset{\bm{\Lambda} \in \mathbb{R}^+}{\min} \sum_{t=1}^T \log(\hat{\mathcal{M}}_t^k(\hat{\bm{\Lambda}}^*) ) - \log(\hat{\mathcal{M}}_t^k(\bm{\Lambda})) \notag \\
& \text{and as } T \rightarrow \infty, \, \hat{\psi}/T \notag \\
& \rightarrow 2 \, \text{KL}\left( \mathcal{M}(\bm{\Lambda}^*) || \mathcal{M}(\bm{\Lambda}_0) \right) +2 \, \text{KL} \left( \mathcal{M}(\bm{\Lambda}_0) || \hat{\mathcal{M}}^k(\bm{\Lambda}_0) \right) \notag \\
& -2 \, \text{KL} \left( \mathcal{M}(\bm{\Lambda}^*) || \hat{\mathcal{M}}^k(\hat{\bm{\Lambda}}^*) \right) -2 \, \underset{\bm{\Lambda} \in \mathbb{R}^+}{\min} \, \text{KL} \left( \hat{\mathcal{M}}^k(\hat{\bm{\Lambda}}^*) || \hat{\mathcal{M}}^k(\bm{\Lambda}) \right) \notag \\
& = \Psi - 2 \left( \text{KL} \left( \mathcal{M}(\bm{\Lambda}^*) || \hat{\mathcal{M}}^k(\hat{\bm{\Lambda}}^*) \right) - \text{KL} \left( \mathcal{M}(\bm{\Lambda}_0) || \hat{\mathcal{M}}^k(\bm{\Lambda}_0) \right) \right) \notag 
\end{flalign}
where $ \Psi = 2 \, \text{KL}\left( \mathcal{M}(\bm{\Lambda}^* || \mathcal{M}(\bm{\Lambda}_0) ) \right) $ and $ \hat{\mathcal{M}}^k(\hat{\bm{\Lambda}}^*) $ is the closest population local maximum at iteration $k$. If as $ k \rightarrow \infty $, the EM model $ \hat{\mathcal{M}}^k $ converges to the true model $\mathcal{M}$ , then $ \hat{\psi}/T \rightarrow \Psi $ 
	
\end{proof}

\section*{Acknowledgment}
This work was supported in part by the Consortium for Verification Technology under Department of Energy National Nuclear Security Administration award number DE-NA0002534, in part by the University of Michigan ECE Departmental Fellow, in part by the U.S. National Science Foundation (NSF) under grant CNS-1737598, and in part by the Southeastern Center for Electrical Engineering Education (SCEEE) under grant SCEEE-17-03.

\ifCLASSOPTIONcaptionsoff
 \newpage
\fi



\bibliographystyle{IEEEtran}
\bibliography{refs}
\end{document}